\documentclass[12pt]{article}
\usepackage{times}
\usepackage{fullpage}
\usepackage[margin=1in]{geometry}
\usepackage{authblk, amsmath, amssymb, amsthm, setspace, bm, mathtools, multirow, makecell, breqn, bigints, enumitem, comment}
\usepackage[normalem]{ulem}
\usepackage{xcolor}
\usepackage{natbib}
\usepackage{hyperref}\hypersetup{colorlinks=true, linkcolor=blue, citecolor=blue, urlcolor=blue}
\usepackage{graphicx, float, caption, subcaption}
\usepackage[title]{appendix}

\providecommand{\keywords}[1]{\small\textbf{Keywords:} #1}

\DeclareMathOperator{\diag}{diag}
\DeclareMathOperator{\ind}{\mathbb{I}}
\newtheorem{theorem}{Theorem}

\newtheorem{proposition}[theorem]{Proposition}

\theoremstyle{remark}
\newtheorem*{remark*}{Remark}

\title{\bf Posterior Mode-Guided Dimension Reduction for Bayesian Model Averaging in Heavy-Tailed Linear Regression}

\author{Shamriddha De\footnote{Postdoctoral Fellow, Real Estate Research Initiative, Indian Institute of Management Bangalore, Bengaluru, India}
\ and
Joyee Ghosh\footnote{Associate Professor, Department of Statistics and Actuarial Science, The University of Iowa, Iowa City, USA}}

\date{}

\begin{document}

\maketitle

\setstretch{1}
\begin{abstract} \noindent
For large model spaces in linear regression with spike-and-slab priors, the potential entrapment of Markov chain Monte Carlo (MCMC)-based methods poses significant challenges in posterior computation. Existing maximum a posteriori (MAP)-based methods provide more computationally viable alternatives, but fail to perform tail heaviness estimation and uncertainty quantification. To address these problems, we propose a method that blends MAP estimation with MCMC-based stochastic search algorithms within an error framework comprising a combination of the hyperbolic and Student-$t$ distributions. The hyperbolic distribution has the light-tailed normal and heavy-tailed Laplace distributions as limiting cases, but is thinner-tailed than the Student-$t$ family. Including the Student-$t$ distribution in the error density enables better adaptation to heavier tails. Amalgamating the two error densities thus ensures a model with more flexible tail behavior when faced with unknown tail thickness in the data, compared to MAP estimators with fixed levels of tail heaviness that assume the errors have a normal or Laplace distribution. Under this proposed error model, the current work develops a two-step expectation conditional maximization (ECM)-guided MCMC algorithm. First, we conduct an ECM-based posterior maximization to guide variable selection. We then execute a Gibbs sampler on the resulting ECM-guided model space for tail heaviness estimation and uncertainty quantification. Through simulation studies and benchmark real datasets, our proposed method is shown to exhibit several advantages in variable selection and uncertainty quantification over state-of-the-art MAP-based methods. To implement our proposed method, we developed the R package \texttt{FlexBayesReg}, available at \href{https://github.com/shamriddha1998/FlexBayesReg}{https://github.com/shamriddha1998/FlexBayesReg}.
\\~\\
\end{abstract}

\onehalfspacing

\section{Introduction} \label{sec:intro}

The extensive literature on variable selection in linear regression models spans a broad range of modeling techniques, ranging from the classical formulations based on normal errors to more recent approaches that emphasize robustness to departures from normality. From a Bayesian point of view, the problem of variable selection is often addressed
by putting spike-and-slab priors on the regression coefficients. Several authors, including \citet{george1993, george1997, gramacy2010, hans2010, hans2011, clyde2011, ghosh2011, ghosh2015, zanella2019, nie2023bbssl}, have utilized such priors to incorporate model uncertainty under normal error models. Spike-and-slab prior formulations have also been adopted for Bayesian variable selection and model averaging in linear regression models with heavy-tailed errors \citep{gramacy2010, ren2023, chakraborty2025, de2026}.

Under spike-and-slab priors, the size of the model space grows exponentially with the number of predictors. Thus, even a modest increase in the number of predictors results in a large model space, whose posterior exploration with MCMC is computationally daunting. Owing to the multimodal structure of the posterior distribution, standard MCMC algorithms are susceptible to entrapment at local modes. Rather than full posterior exploration, a more viable alternative is to focus on the restricted goal of posterior mode estimation, commonly referred to as MAP estimation. The availability of non-stochastic maximization techniques like the expectation maximization (EM) algorithm and its variants \citep{dempster1977, meng1993, wei1990, ruth2024} has further inspired many authors to explore and implement efficient MAP-based techniques. In particular, \citet{rockova2014} proposed the expectation maximization variable selection (EMVS) algorithm with continuous spike-and-slab priors in a normal error model. For the normal likelihood, \cite{rockova2018ssl} later proposed the spike-and-slab lasso by combining spike-and-slab priors with the Bayesian lasso using Laplace priors \citep{park2008} on the regression coefficients. A more robust version of the spike-and-slab lasso with asymmetric Laplace errors was recently developed by \citet{liu2024} through the spike-and-slab quantile lasso (SSQLASSO).

Despite a significant computational gain over MCMC algorithms, deterministic MAP-based methods in the literature have primarily focused on coefficient estimation and variable selection. They lack uncertainty quantification for tail heaviness estimation, which is a core objective of posterior sampling via MCMC methods. In other words, neither MCMC nor MAP-based methods can typically provide a stand-alone pathway to all the aforementioned goals in high-dimensional regression models with potentially heavy-tailed errors. Motivated by this limitation, we try to simultaneously address both computational efficiency and inferential richness within a flexible modeling framework by blending MAP estimation with MCMC-based stochastic search algorithms. For this purpose, we model the errors using an amalgamation of the hyperbolic and the Student-$t$ distributions, denoted by the hyperbolic or Student-$t$ error model (HTEM), and propose a two-step posterior mode-guided MCMC framework. The proposed methodology and its significance are summarized as follows:
\begin{enumerate}
    \item Inspired by \citet{de2026}, we assume that the errors follow either a hyperbolic distribution or a Student-$t$ distribution. An indicator variable $\alpha$ specifies which error distribution is used. The corresponding shape parameter is denoted by $\eta$, and its prior distribution is allowed to depend on $\alpha$. Varying the shape parameter of the hyperbolic distribution encompasses different degrees of tail heaviness, ranging from the light tails of a normal distribution to the heavy tails of a Laplace distribution. Moreover, since the hyperbolic distribution is lighter-tailed than the Student-$t$ family, incorporating the Student-$t$ distribution in our error model allows us to capture a higher degree of tail heaviness. Besides, the Student-$t$ distribution also behaves similar to a normal distribution for large degrees of freedom. In this way, the proposed error density provides greater flexibility in tail behavior than the existing MAP-based techniques like SSQLASSO and EMVS, whose error distributions have Laplace and normal tails, respectively. In contrast, the tail heaviness in our model is estimated from the data.

    \item The first step of our proposed two-step approach guides model space reduction via posterior maximization. For each value of $\alpha$, corresponding to either the hyperbolic or Student-$t$ error distribution, we fix $\eta$ at a suitably chosen value and perform posterior maximization using the ECM algorithm. We show that each of the CM updates in our ECM algorithm attains the unique global maximum for that parameter block, given other parameters and quantities in the E-step. This in turn ensures that the target objective function in the ECM algorithm, that is the log-posterior, is indeed nondecreasing at every step of the ECM algorithm. The tail thickness is treated as fixed at this stage as ECM-based estimation of tail heaviness is often inaccurate. Thereafter, the suitable value of $\alpha$ and the corresponding set of parameter estimates is chosen by maximizing the log-posterior. This ECM-guided step provides a computationally efficient way to construct a reduced model space from the selected predictors, which is subsequently explored by the Gibbs sampler.

    \item The second step is based on a slightly modified version of the stochastic search variable selection approach of \citet{de2026} under HTEM. More specifically, under HTEM, we implement a Gibbs sampler with a Metropolis-Hastings step on the reduced model space. We assign a Bernoulli prior to the error distribution indicator $\alpha$, and perform posterior inference on tail heaviness by assigning a discrete uniform prior to the shape parameter $\eta$. However, unlike \citet{de2026}, we assign a prior to $\eta$ conditional on $\alpha$, using separate supports for $\eta$ under the hyperbolic and Student-$t$ error distributions, to have more flexibility in modeling tail heaviness.
\end{enumerate}
We use the shorthand GECM-HTEM to denote our proposed ECM-guided Gibbs (G) sampling algorithm for the combined error model with hyperbolic and Student-$t$ distributions (HTEM).

The rest of the article is organized as follows. In Section \ref{sec:gecm}, we introduce our proposed two-step GECM-HTEM algorithm. In Section \ref{sec:simul}, we carry out simulation studies to compare the performance with two other state-of-the-art MAP estimation techniques. We analyze three real life datasets in Section \ref{sec:real}. Concluding remarks and possible directions for future work are stated in Section \ref{sec:conc}. The proof of a proposition in Section \ref{sec:gecm_gecm-htem_ecm} pertaining to the normal scale mixture representation of the hyperbolic distribution is outlined in Appendix \ref{app:scalemix_hyperbolic}. In Appendix \ref{app:ecm-htem_max}, we prove a proposition to justify unique global maximizations of the respective conditional objective functions by the continuous CM updates in the ECM step in Section \ref{sec:gecm_gecm-htem_ecm}. Some additional results for real data analysis are reported in Appendix \ref{app:real}. For clarity on parameterization, the probability density functions of some distributions used in the article are provided in Appendix \ref{app:densities}.

\section{The GECM Approach Towards the Mixture of Hyperbolic and Student-$t$ Error Models} \label{sec:gecm}

In this section, we present our proposed two-step posterior mode-guided MCMC technique. In Section \ref{sec:gecm_model}, we begin by formulating our model and briefly highlight the pitfalls of MCMC to motivate our proposed method. In Section \ref{sec:gecm_gecm-htem_ecm}, we introduce slightly modified normal scale mixture representations of the two error models, which facilitate closed-form updates in the conditional maximization steps of the ECM algorithm. Since the normal scale mixture representation of the Student-$t$ error model is well known, Proposition \ref{prop:scalemix_hyperbolic} provides a proof only for the corresponding representation of the hyperbolic error model. We then describe the steps of the ECM algorithm and show in Proposition \ref{prop:ecm-htem_max} that, with the remaining quantities held fixed, each conditional maximization update attains its unique global maximum in the interior of the parameter space. In Section \ref{sec:gecm_gecm-htem_g}, we describe the Gibbs sampling step for uncertainty quantification of the parameters. We close the methodological development in Section  \ref{sec:gecm_hyperpar} with a detailed discussion of the choice of hyperparameters used in the ECM and Gibbs sampling steps.

\subsection{Model Formulation} \label{sec:gecm_model}

To induce flexibility in tail heaviness, \citet{de2026} proposed a mixture of the hyperbolic and the Student-$t$ distributions (HTEM) for modeling the errors. A slight variant of such an error model is used in the present paper, with the differences described later in this section. We consider a linear regression model
\begin{equation}
    \bm{Y} = X\bm{\beta} + \bm{\epsilon},
\label{eq:regmodel}
\end{equation}
with an $n \times 1$ response vector $\bm{Y}$, an $n \times p$ design matrix $X$ of $p$ covariates, $\bm{\beta}$ as the $p \times 1$ vector of regression coefficients, and $\bm{\epsilon}$ as the $n$-dimensional vector of independent and identically distributed errors. Including an intercept term in \eqref{eq:regmodel} is redundant, since the covariates and the response variables are centered and scaled about their respective means and standard deviations. We assume that all the errors follow either a hyperbolic or a Student-$t$ distribution as follows:
\begin{equation}
    \epsilon_i \mid \alpha, \eta, \rho^2
    \stackrel{\mathrm{iid}}{\sim}
    \mathrm{Hyperbolic}(\eta, \rho^2) \ind{(\alpha = 0)}
    + \mathrm{t}(\eta, \rho^2) \ind{(\alpha = 1)},
    \quad
    i = 1, 2, \dots, n,
\label{eq:htem_errdist}
\end{equation}
where $\mathrm{Hyperbolic}(\eta, \rho^2)$ signifies a Hyperbolic distribution with shape and scale parameters $\eta$ and $\rho^2$, respectively, $\mathrm{t}(\eta, \rho^2)$ refers to a Student-$t$ distribution with location parameter zero, degrees of freedom parameter $\eta$ and scale parameter $\rho^2$, and the parameter $\alpha \in \{0, 1\}$ denotes an indicator variable for choosing the error distribution between hyperbolic and Student-$t$. Equation \eqref{eq:htem_errdist} implies that conditional on $\alpha = 0$ and $\alpha = 1$, all the errors are generated from the hyperbolic and the Student-$t$ densities, respectively, that is,
\begin{equation}
    \epsilon_i \mid \alpha=0, \eta, \rho^2 \stackrel{\mathrm{iid}}{\sim} \mathrm{Hyperbolic}(\eta, \rho^2),
    \quad
    \epsilon_i \mid \alpha=1, \eta, \rho^2 \stackrel{\mathrm{iid}}{\sim} \mathrm{t}(\eta, \rho^2),
    \quad
    i = 1, 2, \dots, n.
\label{eq:htem_errdist-special}
\end{equation}
The parameters $\eta$ and $\rho^2$ serve as the tail heaviness and the scale parameters of the error densities, respectively. The model uncertainty in \eqref{eq:regmodel} is introduced using a $p$-dimensional latent binary vector $\bm{\gamma} = (\gamma_1, \gamma_2, \dots, \gamma_p)^\top$, where $\gamma_j \in \{0, 1\}$, for every $j = 1, 2, \dots, p$, with $\gamma_j = 0$ and $\gamma_j = 1$ respectively signifying the exclusion or inclusion of the $j$th covariate in the model. For every model $\bm{\gamma}$ in the model space $\bm{\Gamma} = \{0, 1\}^p$, we define $\bm{\beta}_{\bm{\gamma}}$ to be the $p_{\bm{\gamma}}$-dimensional vector of the non-zero components in $\bm{\beta}$, where $p_{\bm{\gamma}} = \sum_{j=1}^{p} \gamma_j$, and let $X_{\bm{\gamma}}$ denote the matrix of the corresponding columns of the design matrix $X$.

\citet{gneiting1997} expressed the hyperbolic density as a normal scale mixture with a generalized inverse Gaussian (GIG) mixing density. Similarly, normal scale mixture representations of the Student-$t$ distribution with an inverse gamma mixing density are extensively available in the literature. Both the normal scale mixtures have been utilized by \citet{de2026} to formulate a computationally efficient likelihood for model \eqref{eq:regmodel}. The likelihood and the hierarchical prior specification for the unknown parameters are stated below:
\begin{align}
    \bm{Y} \mid \bm{\beta}, \bm{\sigma}
    & \sim
    \mathrm{N}(X\bm{\beta}, \Sigma),
\nonumber \\
    \sigma^2_i \mid \alpha, \eta, \rho^2
    & \stackrel{\mathrm{iid}}{\sim}
    \mathrm{GIG}(1, \eta/\rho^2, \eta\rho^2) \ind{(\alpha = 0)}
    + \mathrm{InvGamma}(\eta/2, \eta\rho^2/2) \ind{(\alpha = 1)}, \nonumber \\
    & \quad i = 1, 2, \dots, n,
\nonumber \\
    \beta_j \mid \gamma_j, \rho^2, \tau^2
    & \stackrel{\mathrm{ind}}{\sim}
    \mathrm{N}(0, \rho^2\tau^2)\ind{(\gamma_j = 1)} + \delta_{\{0\}}(\beta_j)\ind{(\gamma_j = 0)},
    \quad j = 1, 2, \dots, p,
\nonumber \\
    \tau^2
    & \sim
    \mathrm{InvGamma}(\lambda_{\tau}/2, \lambda_{\tau}/2),
\nonumber \\
    \gamma_j \mid \theta
    & \stackrel{\mathrm{iid}}{\sim}
    \mathrm{Bernoulli}(\theta),
    \quad j = 1, 2, \dots, p,
\nonumber \\
    \theta
    & \sim
    \mathrm{Beta}(c_{\theta}, d_{\theta}),
\nonumber \\
    \rho^2
    & \sim
    \mathrm{InvGamma}(a_{\rho}, b_{\rho}),
\nonumber \\
    \eta \mid \alpha
    & \sim
    \mathrm{DiscUniform}(\mathcal{G}_{\eta}^{0}) \ind{(\alpha = 0)}
    + \mathrm{DiscUniform}(\mathcal{G}_{\eta}^{1}) \ind{(\alpha = 1)},
\nonumber \\
    \alpha \mid \omega
    & \sim
    \mathrm{Bernoulli}(\omega),
\nonumber \\
    \omega
    & \sim
    \mathrm{Beta}(r_{\omega}, s_{\omega}),
\label{eq:g-htem_lik-priors}
\end{align}
where $\bm{\sigma} = (\sigma^2_1, \sigma^2_2, \dots, \sigma^2_n)^\top$, $\Sigma = \diag{(\sigma^2_1, \sigma^2_2, \dots, \sigma^2_n)}$, $\lambda_{\tau} > 0$, $a_{\rho} > 0$, $b_{\rho} > 0$, $c_{\theta} > 0$, $d_{\theta} > 0$, $r_{\omega} > 0$, $s_{\omega}>0$, $\mathcal{G}_{\eta}^{0}$ and $\mathcal{G}_{\eta}^{1}$ are finite sets of positive real numbers, $\delta_{\{0\}}(\cdot)$ represents a point mass at zero, and $\ind{(\cdot)}$ denotes the indicator function. At this juncture, it is worthwhile to mention that our proposed prior specification in \eqref{eq:g-htem_lik-priors} differs from that of \citet{de2026} in the treatment of the tail heaviness parameter $\eta$. \citet{de2026} used a common support for the discrete uniform prior on $\eta$, irrespective of the error indicator $\alpha$. In contrast, we consider separate supports for $\eta$ in \eqref{eq:g-htem_lik-priors}, conditional on $\alpha$, to allow greater flexibility in tail thickness while preserving the moment conditions needed for standardization of the response variables. In particular, we recommend $\mathcal{G}_{\eta}^{0} = \{0.05, 0.1, 0.2, 0.3, \dots, 0.9, 1, 2, 5, 10, 20, 50\}$ and $\mathcal{G}_{\eta}^{1} = \{2.1, 5, 10, 20, 50\}$, given $\alpha = 0$ (hyperbolic) and $\alpha = 1$ (Student-$t$), respectively, as suitable grids encompassing wide ranges of tail heaviness. These choices ensure the existence of the first two moments of the error distribution, thereby allowing the usual standardization of the response variable using the mean and standard deviation. The restriction ensuring the existence of the first two moments is also empirically motivated. When the prior support for $\eta$ includes values for which these moments are not guaranteed to exist, standardization based on the mean and standard deviation is not theoretically justified. In our numerical experiments, using a corresponding robust standardization based on the median and quartile deviation, as in \citet{de2026}, led to less satisfactory performance in large $p$ settings.

The priors and likelihood in \eqref{eq:g-htem_lik-priors} enable fully Bayesian inference through a Gibbs sampler with a Metropolis-Hastings step for posterior computation. While this setup can be viewed as a gold standard for Bayesian model averaging in a heavy-tailed framework, computational challenges emerge in terms of scalability to large model spaces. The multimodal nature of the posterior distribution over the models often causes the Gibbs sampler to get trapped in local modes. Escaping such entrapment demands the execution of extremely long chains of the sampler and compromises the computational efficiency of the algorithm. Resorting to a simplified goal of MAP estimation through EM-like algorithms appears to be a more practical alternative to address these computational difficulties. However, the lack of uncertainty quantification in such deterministic techniques highlights the continued need for the Gibbs sampler. Accordingly, we propose an algorithm that combines the computational efficiency of deterministic MAP estimation with the uncertainty quantification provided by stochastic search MCMC.

\subsection{The ECM-Guided Gibbs Sampling} \label{sec:gecm_gecm-htem}

We develop a two-step method that combines deterministic MAP estimation for computational speed with stochastic search MCMC for uncertainty quantification, particularly in the tails, thereby leveraging the complementary strengths of both the approaches. Our proposed technique begins with an initial level of variable selection based on posterior mode estimation via the ECM algorithm \citep{meng1993}, followed by a Gibbs sampler that performs a final level of variable selection on the reduced model space obtained from the modal estimates.

\subsubsection{The ECM Step} \label{sec:gecm_gecm-htem_ecm}

In the first step, we intend to estimate the mode of the posterior distribution, marginalized over the model uncertainty parameter $\bm{\gamma}$, that is,
\begin{equation}
    p(\bm{\beta}, \rho^2, \eta, \alpha, \theta, \omega, \tau^2, \bm{\sigma} \mid \bm{Y}) = \sum_{\bm{\gamma}} p(\bm{\gamma}, \bm{\beta}, \rho^2, \eta, \alpha, \theta, \omega, \tau^2, \bm{\sigma} \mid \bm{Y}).
\end{equation}
However, for a large value of $p$, direct maximization through complete enumeration of all the $2^p$ models and taking summation over them is computationally infeasible. Therefore, following \citet{rockova2014}, we consider $\bm{\gamma}$ as the ``missing data" and maximize the conditional expectation of the ``complete data" log-posterior $\log{p(\bm{\gamma}, \bm{\beta}, \rho^2, \eta, \alpha, \theta, \omega, \tau^2, \bm{\sigma} \mid \bm{Y})}$. The conditional expectation is taken over $\bm{\gamma}$, given the observed data $\bm{Y}$ and the current parameter estimates. Due to the unavailability of closed-form expressions for simultaneous maximization for all the parameters, the ECM algorithm \citep{meng1993} is applied to maximize over blocks of parameters conditionally on others. For computational ease in the ECM step, we make a few changes to \eqref{eq:g-htem_lik-priors}. 

First, we deploy slightly modified versions of the scale mixture representations of the hyperbolic and Student-$t$ densities in the ECM step. In particular, the modified representation for the hyperbolic density is guided by Proposition \ref{prop:scalemix_hyperbolic}, which is proven in Appendix \ref{app:scalemix_hyperbolic}.
\begin{proposition}
    For fixed $\rho^2 > 0$ and $\eta > 0$, let $A$ and $a$ be random variables such that $A \mid a^2, \rho^2 \sim \mathrm{N}(0, \rho^2a^2)$ and $a^2 \mid \eta \sim \mathrm{GIG}(1, \eta, \eta)$. Then $A \mid \eta, \rho^2$ is distributed as $\mathrm{Hyperbolic}(\eta, \rho^2)$.
\label{prop:scalemix_hyperbolic}
\end{proposition}

Second, to ensure differentiability of the objective functions, we replace the mixture of point mass at zero and a continuous prior (point mass spike-and-slab prior) on the regression coefficients by a continuous spike-and-slab prior. That is, we consider a mixture of two continuous distributions as the prior, where one component (spike) is concentrated around zero, while the other component (slab) is more diffuse. Such a prior shrinks the coefficients for the noise variables towards zero, but cannot enforce them to be exactly zero. Therefore, the modal estimate of $\bm{\beta}$ does not explicitly drop any covariate from the model. 

Third, preliminary analyses of the ECM algorithm revealed poor estimation of the tail heaviness for modest sample size or weak signal strength. Furthermore, \citet{fernandez1999} pointed out the potential asymptotic unboundedness of the Student-$t$ likelihood, and cautioned against using likelihood maximization techniques such as EM-type algorithms for this distribution when the degrees of freedom are small. For these reasons, in the ECM step we do not maximize over $\eta$. Instead, we fix $\eta$ conditionally on the error indicator $\alpha$, setting $\eta = \eta_\alpha$ for $\alpha \in \{0,1\}$, and perform separate ECM-based posterior maximizations for the two component error models. We define the optimal value of $\alpha$ as the value in $\{0,1\}$ that yields the larger maximized log-posterior. More specifically, we use $\eta = \eta_0 = 1$ and $\eta = \eta_1 = 4.1$ corresponding to the hyperbolic ($\alpha = 0$) and the Student-$t$ ($\alpha = 1$) error models, respectively. In the case of hyperbolic errors, $\eta = \eta_0 = 1$ yields tails that are somewhat heavier than those of the normal distribution but considerably lighter than those of the Laplace distribution. Similarly, choosing $\eta = \eta_1 = 4.1$ under the Student-$t$ error model allows for tails heavier than those of the normal distribution while ensuring the existence of the first four moments. This strategy tends to provide some safeguard against the potential adversities of variable selection under a misspecified light-tailed model. For the parameters $\bm{\gamma}$, $\rho^2$, $\tau^2$ and $\theta$, we retain the priors in \eqref{eq:g-htem_lik-priors}. In a nutshell, using Proposition \ref{prop:scalemix_hyperbolic}, the likelihood as well as the priors used for the ECM step are stated below:
\begin{align}
    \bm{Y} \mid \bm{\beta}, \rho^2, \bm{\sigma}
    & \sim
    \mathrm{N}(X\bm{\beta}, \rho^2\Sigma),
\nonumber \\
    \sigma^2_i \mid \alpha, \eta = \eta_{\alpha}
    & \stackrel{\mathrm{iid}}{\sim}
    \mathrm{GIG}(1, \eta_0, \eta_0) \ind{(\alpha = 0)}
    + \mathrm{InvGamma}(\eta_1/2, \eta_1/2) \ind{(\alpha = 1)}, \nonumber \\
    & \quad i = 1, 2, \dots, n,
\nonumber \\
    \beta_j \mid \gamma_j, \rho^2, \tau^2
    & \stackrel{\mathrm{ind}}{\sim}
    (1-\gamma_j)\mathrm{N}(0, \kappa_0\rho^2\tau^2) + \gamma_j\mathrm{N}(0, \kappa_1\rho^2\tau^2),
    \quad j = 1, 2, \dots, p,
    \nonumber
\\
    \tau^2
    & \sim
    \mathrm{InvGamma}(\lambda_{\tau}/2, \lambda_{\tau}/2),
    \nonumber
\\
    \rho^2
    & \sim
    \mathrm{InvGamma}(a_{\rho}, b_{\rho}),
    \nonumber
\\
    \gamma_j \mid \theta
    & \stackrel{\mathrm{ind}}{\sim}
    \mathrm{Bernoulli}(\theta),
    \quad j = 1, 2, \dots, p,
    \nonumber
\\
    \theta
    & \sim
    \mathrm{Beta}(c_{\theta}, d_{\theta}),
\label{eq:ecm-htem_lik-priors}
\end{align}
where $\eta_0 = 1$, $\eta_1 = 4.1$, and $\kappa_0 > 0$ and $\kappa_1 > 0$ are respectively chosen to be small and large, corresponding to the spike and the slab components. The hyperparameters $\kappa_0$ and $\kappa_1$ control the strengths of the respective spike and slab components. A smaller value of $\kappa_0$ indicates a more peaked spike with more concentration around zero, while a larger value of $\kappa_1$ signifies a flatter slab.

Let $\mathbb{E}^*_{\bm{\gamma}}[\cdot]$ denote the conditional expectation $\mathbb{E}_{\bm{\gamma}}[\cdot \mid \bm{Y}, \alpha, \eta = \eta_{\alpha}, \hat{\bm{\beta}}, \hat{\rho}^2, \hat{\theta}, \hat{\tau}^2, \hat{\bm{\sigma}}]$ over $\bm{\gamma}$, given $\bm{Y}$, $\alpha$, $\eta = \eta_{\alpha}$, and the current estimates $\hat{\bm{\beta}}$, $\hat{\rho}^2$, $\hat{\theta}$, $\hat{\tau}^2$ and $\hat{\bm{\sigma}}$ of the remaining parameters $\bm{\beta}$, $\rho^2$, $\theta$, $\tau^2$ and $\bm{\sigma}$, respectively. The goal is to separately maximize the conditional expectation $\mathbb{E}^*_{\bm{\gamma}}[\log{p(\bm{\gamma}, \bm{\beta}, \rho^2, \theta, \tau^2, \bm{\sigma})}]$ for each value of $\alpha$ ($\alpha = 0, 1$), which, using the likelihood and the priors in \eqref{eq:ecm-htem_lik-priors}, reduces to iteratively maximizing the function
\begin{equation}
    Q^{\alpha} = Q_{\bm{Y}} + Q_{\bm{\sigma}}^{\alpha} + Q_{\bm{\beta}} + Q_{\tau^2} + Q_{\rho^2} + Q_{\bm{\gamma}} + Q_{\theta},
    \quad \alpha = 0, 1,
\end{equation}
where the summands can be expanded as
\begin{align}
    Q_{\bm{Y}}
    & = - \dfrac{n}{2}\log{\rho^2} - \frac{1}{2}\sum_{i=1}^n \log{\sigma^2_i} - \dfrac{1}{2\rho^2} (\bm{Y} - X\bm{\beta})^\top \Sigma^{-1} (\bm{Y} - X\bm{\beta}),
\nonumber \\
    Q_{\bm{\sigma}}^{\alpha}
    & =
    \left[ - \dfrac{\eta_0}{2} \sum_{i=1}^n \left(\sigma^2_i + \dfrac{1}{\sigma^2_i}\right) \right] \ind{(\alpha = 0)}
    +
    \left[ - \left(\dfrac{\eta_1}{2}+1\right) \sum_{i=1}^n \log{\sigma^2_i} - \dfrac{\eta_1}{2} \sum_{i=1}^n \dfrac{1}{\sigma^2_i} \right] \ind{(\alpha = 1)},
\nonumber \\
    Q_{\bm{\beta}}
    & = - \dfrac{1}{2} \sum_{j=1}^p \mathbb{E}^*_{\bm{\gamma}}[\log{v_j}] - \dfrac{p}{2}\log{\rho^2} - \dfrac{p}{2}\log{\tau^2} - \dfrac{1}{2\rho^2\tau^2} \bm{\beta}^\top \Psi \bm{\beta},
\nonumber \\
    Q_{\tau^2}
    & = - \left(\dfrac{\lambda_{\tau}}{2}+1\right)\log{\tau^2} - \dfrac{\lambda_{\tau}}{2\tau^2},
\nonumber \\
    Q_{\rho^2}
    & = - (a_{\rho}+1)\log{\rho^2} - \dfrac{b_{\rho}}{\rho^2},
\nonumber \\
    Q_{\bm{\gamma}}
    & = (\log{\theta}) \sum_{j=1}^p \mathbb{E}^*_{\bm{\gamma}}[\gamma_j] + (\log{(1-\theta)}) \sum_{j=1}^p \mathbb{E}^*_{\bm{\gamma}}[1-\gamma_j],
\nonumber \\
    Q_{\theta}
    & = (c_{\theta}-1)\log{\theta} + (d_{\theta}-1)\log{(1-\theta)},
\label{eq:ecm-htem_objective}
\end{align}
with $\eta_0 = 1$, $\eta_1 = 4.1$, $v_j = \kappa_0(1-\gamma_j) + \kappa_1\gamma_j$ for $j = 1, 2, \dots, p$, and $\Psi = \diag((\mathbb{E}^*_{\bm{\gamma}}[1/v_j]))$. The complicated nature of the objective functions $Q^0$ and $Q^1$ prohibits closed-form maximization to be performed jointly over all the parameters using the traditional EM algorithm \citep{dempster1977}. As an alternative, we adopt the ECM algorithm of \citet{meng1993}. For each fixed value of $\alpha\in \{0,1\}$, we run a separate ECM algorithm by sequentially maximizing the corresponding objective function $Q^\alpha$ with respect to one parameter (or block of parameters), conditional on the current estimates of the remaining parameters under that value of $\alpha$. Thus, the ECM procedure is carried out separately for the hyperbolic and Student-$t$ error models, and the resulting maximized log-posteriors (marginalized over $\bm{\sigma}$) are subsequently compared to select the optimal value of $\alpha$. An outline of the ECM algorithm, followed by the variable selection step, is provided below.
\begin{enumerate}
\item \textbf{The E-Step:}
In the expectation (E) step, we compute the conditional expectations $\mathbb{E}^*_{\bm{\gamma}}[\gamma_j]$, $\mathbb{E}^*_{\bm{\gamma}}[\log{v_j}]$ and $\mathbb{E}^*_{\bm{\gamma}}[1/v_j]$, for every $j = 1, 2, \dots, p$. We note that the conditional posterior distribution of $\bm{\gamma}$, given the response variables $\bm{Y}$, $\alpha$, $\eta = \eta_{\alpha}$ and the current estimates $(\hat{\bm{\beta}}, \hat{\rho}^2, \hat{\theta}, \hat{\tau}^2, \hat{\bm{\sigma}})$, depends on $(\bm{Y}, \alpha, \eta=\eta_{\alpha}, \bm{\sigma})$ only through $(\hat{\bm{\beta}}, \hat{\rho}^2, \hat{\theta}, \hat{\tau}^2)$, and the conditional expectation $\mathbb{E}^*_{\bm{\gamma}}[\gamma_j]$ simplifies to 
\begin{align}
    \mathbb{E}^*_{\bm{\gamma}}[\gamma_j]
    & = \mathbb{P}(\gamma_j = 1 \mid \bm{Y}, \alpha, \eta = \eta_{\alpha}, \hat{\bm{\beta}}, \hat{\rho}^2, \hat{\theta}, \hat{\tau}^2, \hat{\bm{\sigma}})
\nonumber \\
    & = \dfrac{p(\hat{\beta}_j \mid \gamma_j = 1, \hat{\rho}^2, \hat{\tau}^2) \mathbb{P}(\gamma_j = 1 \mid \hat{\theta})} {\sum_{t=0}^1 p(\hat{\beta}_j \mid \gamma_j = t, \hat{\rho}^2, \hat{\tau}^2) \mathbb{P}(\gamma_j = t \mid \hat{\theta})}
    = g_j,
\label{eq:ecm-htem_estep1}
\end{align}
for every $j = 1, 2, \dots, p$. The prior specification in \eqref{eq:ecm-htem_lik-priors} ensures computational tractability of the quantities involved in \eqref{eq:ecm-htem_estep1}. Furthermore, for every $j = 1, 2, \dots, p$, the remaining conditional expectations can be evaluated using \eqref{eq:ecm-htem_estep1} as weighted averages of functions of the spike and the slab strength hyperparameters $\kappa_0$ and $\kappa_1$ in the following manner:
\begin{align}
    \mathbb{E}^*_{\bm{\gamma}}[\log{v_j}]
    & = \mathbb{E}^*_{\bm{\gamma}}[\log{\{\kappa_0(1-\gamma_j) + \kappa_1\gamma_j\}}]
    = (1-g_j)\log{\kappa_0} + g_j\log{\kappa_1},
    \nonumber
\\
    \mathbb{E}^*_{\bm{\gamma}}[1/v_j]
    & = \mathbb{E}^*_{\bm{\gamma}}\left[\dfrac{1}{\kappa_0(1-\gamma_j) + \kappa_1\gamma_j}\right]
    = \dfrac{1-g_j}{\kappa_0} + \dfrac{g_j}{\kappa_1}.
\label{eq:ecm-htem_step2}
\end{align}

\item \textbf{The CM-Step:}
The conditional maximization (CM) step of the ECM algorithm maximizes the objective function $Q^{\alpha}$ for each parameter (or block of parameters), given the current estimates of the remaining parameters, and yields the following closed-form estimates at every iteration:
\begin{align}
    \hat{\bm{\beta}}
    & = \left(X^\top \hat{\Sigma}^{-1} X + \dfrac{1}{\hat{\tau}^2} \Psi\right)^{-1} X^\top \hat{\Sigma}^{-1} \bm{Y},
    \nonumber
\\
    \hat{\rho}^2
    & = \dfrac{2b_{\rho} + (\bm{Y} - X\hat{\bm{\beta}})^\top \hat{\Sigma}^{-1} (\bm{Y} - X\hat{\bm{\beta}}) + \hat{\bm{\beta}}^\top \Psi \hat{\bm{\beta}} / \hat{\tau}^2}{n + p + 2a_{\rho} + 2},
    \nonumber
\\
    \hat{\tau}^2
    & = \dfrac{\lambda_{\tau} + \hat{\bm{\beta}}^\top \Psi \hat{\bm{\beta}} / \hat{\rho}^2}{p + \lambda_{\tau} + 2},
    \nonumber
\\
    \hat{\theta}
    & = \dfrac{c_{\theta} + \sum_{j=1}^p g_j - 1}{c_{\theta} + d_{\theta} + p - 2},
    \nonumber
\\
    \hat{\sigma}^2_i
    & = \dfrac{1}{2\eta_0} \left[ -1 + \left\{1 + 4\eta_0\left(\eta_0 + \dfrac{(y_i - \bm{x}_i^\top\hat{\bm{\beta}})^2}{\hat{\rho}^2}\right)\right\}^{1/2} \right] \ind{(\alpha = 0)} \nonumber \\
    & \quad +
    \dfrac{1}{\eta_1 + 3} \left[ \eta_1 + \dfrac{(y_i - \bm{x}_i^\top\bm{\beta})^2}{\rho^2} \right] \ind{(\alpha = 1)},
    \quad i = 1, 2, \dots, n,
\label{eq:ecm-htem_cmstep}
\end{align}
where $\eta_0 = 1$, $\eta_1 = 4.1$, $\hat{\Sigma} = \diag{(\hat{\sigma}^2_1, \hat{\sigma}^2_2, \dots, \hat{\sigma}^2_n)}$, and for every $i = 1, 2, \dots, n$, $\bm{x}_i$ is the $i$th row of the design matrix $X$ and $y_i$ is the $i$th component of the response vector $\bm{Y}$. We establish in Proposition \ref{prop:ecm-htem_max}  that the continuous CM updates in \eqref{eq:ecm-htem_cmstep} are well-defined and correspond to unique global maximizers of the respective conditional objective functions. Its proof is provided in Appendix \ref{app:ecm-htem_max}.
\begin{proposition}
    Under the likelihood and prior specification in \eqref{eq:ecm-htem_lik-priors}, let us assume that $\kappa_0 > 0$, $\kappa_1 > 0$, $a_{\rho} > 0$, $b_{\rho} > 0$, $\lambda_{\tau} > 0$, $c_\theta > 1$ and $d_\theta > 1$. For each fixed error model, namely the hyperbolic error model with $\alpha = 0$ and $\eta = \eta_0$, and the Student-$t$ error model with $\alpha = 1$ and $\eta = \eta_1$, conditional on the quantities computed in the E-step and on the current values of the remaining parameters, each continuous CM-step with respect to $\bm{\beta}$, $\rho^2$, $\tau^2$, $\theta$, $\sigma^2_i$ ($i = 1, 2, \dots, n$), has a unique global maximizer in the interior of the corresponding parameter space. Moreover, the unique maximizer in each case is attained by the corresponding closed-form update in \eqref{eq:ecm-htem_cmstep}.
\label{prop:ecm-htem_max}
\end{proposition}
\begin{remark*}
For each fixed $\alpha$ and $\eta = \eta_{\alpha}$, Proposition \ref{prop:ecm-htem_max} verifies the blockwise global maximization condition for the continuous CM updates. Hence, by the standard ECM monotonicity argument, the log-posterior marginalized over $\gamma$, that is, 
$\log{p(\beta, \rho^2, \theta, \tau^2, \bm{\sigma} \mid \bm{Y}, \alpha, \eta = \eta_{\alpha})}$, is nondecreasing along the ECM iterations.
\end{remark*}

\item \textbf{Variable Selection:}
Under continuous spike-and-slab priors, the modal estimates of the regression coefficients cannot be exactly zero and therefore do not allow explicit exclusion of a covariate. Moreover, the median probability model \citep{barbieri2004} and other Bayes factor based approaches \citep{ghosh2015, de2026}, which rely on thresholding the posterior probability $\mathbb{P}(\gamma_j = 1 \mid \bm{Y})$ to select covariates, are not applicable here because the posterior distribution over models is unavailable. \citet{rockova2014} estimated the most probable submodel as the highest probability model conditional on the modal estimates of the parameters. Following this strategy, we choose our best submodel as the model with the highest conditional probability, given $\alpha$, $\eta = \eta_{\alpha}$ and the modal estimates $(\hat{\bm{\beta}}, \hat{\rho}^2, \hat{\theta}, \hat{\tau}^2, \hat{\bm{\sigma}})$ from the ECM algorithm. That is, the estimated model is
\begin{equation}
    \hat{\bm{\gamma}}
    = \arg \max_{\bm{\gamma}} \mathbb{P}(\bm{\gamma} \mid \alpha, \eta=\eta_{\alpha}, \hat{\bm{\beta}}, \hat{\rho}^2, \hat{\theta}, \hat{\tau}^2, \hat{\bm{\sigma}}),
\end{equation}
which can be obtained by thresholding the conditional probability of each component as
\begin{equation}
    \hat{\gamma}_j = 1
    \quad \iff \quad
    \mathbb{P}(\gamma_j = 1 \mid \alpha, \eta=\eta_{\alpha}, \hat{\bm{\beta}}, \hat{\rho}^2, \hat{\theta}, \hat{\tau}^2, \hat{\bm{\sigma}}) \geq 0.5,
\end{equation}
for every $j = 1, 2, \dots, p$. In other words, the $j$th covariate is dropped from our model if the conditional probability of the corresponding ``slab" component, given $\alpha$, $\eta = \eta_{\alpha}$ and the modal estimates $(\hat{\bm{\beta}}, \hat{\rho}^2, \hat{\theta}, \hat{\tau}^2, \hat{\bm{\sigma}})$, falls below a threshold of 0.5.
\end{enumerate}
We execute the ECM algorithm along with the variable selection separately for the hyperbolic (with $\alpha = 0$ and $\eta = \eta_0 = 1$) and the Student-$t$ ($\alpha = 1$ and $\eta = \eta_0 = 1$) error models thereby yielding two sets of estimates of the parameters. Between the two, the set that maximizes the log-posterior marginalized over $\bm{\sigma}$ and evaluated at the remaining estimates is chosen to be the optimal set of estimates.

The estimated inclusion indicator corresponding to the optimal set of estimates defines an ECM-guided model space, consisting of $2^{p^*}$ models, where $p^* \leq p$, based on the design matrix $X^*$ with $p^*$ columns from the original design matrix $X$ corresponding to the included covariates. When the true data generating model contains substantially fewer signals than $p$, this ECM-guided model space is intended to improve the mixing and stability of the Gibbs sampler compared to sampling in the original model space. Accordingly, this model space is utilized in the Gibbs sampling step of our proposed algorithm. 

\subsubsection{The Gibbs Sampling Step} \label{sec:gecm_gecm-htem_g}

In the second step, we study the tail behavior and perform uncertainty quantification through an MCMC algorithm on the ECM-guided model space deduced from the ECM-based modal estimates. The likelihood and priors in \eqref{eq:g-htem_lik-priors} corresponding to the stochastic search HTEM enable closed-form full conditional distributions, which can be further utilized to deduce an efficient Gibbs sampler, along with an MH-step to address model selection. Although a preliminary level of model space reduction is performed in the ECM step, we incorporate this additional step of model uncertainty on the ECM-guided model space to allow dropping of any potential false positives/signals. Moreover, unlike in the ECM step, we put a Bernoulli prior on the error distribution indicator $\alpha$, and discrete uniform priors on the tail heaviness parameter $\eta$, with separate supports conditional on $\alpha$, as indicated in \eqref{eq:g-htem_lik-priors}. This prior specification induces flexibility in tail behavior through seamless adaptivity between the component error distributions (determined by $\alpha$) as well as different degrees of tail thickness (regulated by $\eta$) within each component. The resulting conditional distributions for posterior computation are stated as follows:
\begin{align}
    p(\bm{\gamma} & \mid \bm{Y}, \bm{\sigma}, \tau^2, \rho^2, \theta) \nonumber \\
    & \propto
    \left[
        \dfrac{|D_{\bm{\gamma}}|^{-1/2}}{(\tau^2 \rho^2)^{p_{\bm{\gamma}}/2}}
        \exp{\left( \dfrac{1}{2} \bm{Y}^\top \Sigma^{-1} X_{\bm{\gamma}} D_{\bm{\gamma}}^{-1} X_{\bm{\gamma}}^\top \Sigma^{-1} \bm{Y} \right)}
        \ind(\bm{\gamma} \neq \bm{0})
        + \ind(\bm{\gamma} = \bm{0})
    \right] \nonumber \\
    & \quad \times \theta^{p_{\bm{\gamma}}} (1 - \theta)^{p - p_{\bm{\gamma}}}
\nonumber \\
    \bm{\beta} & \mid \bm{Y}, \bm{\gamma}, \bm{\sigma}, \tau^2, \rho^2
    \sim
    \mathrm{N} \left(
        D_{\bm{\gamma}}^{-1} X_{\bm{\gamma}}^\top \Sigma^{-1} \bm{Y},
        D_{\bm{\gamma}}^{-1},
    \right)
\nonumber \\
    \rho^2 & \mid \bm{\gamma}, \bm{\beta}, \bm{\sigma}, \tau^2, \alpha, \eta \nonumber \\
    & \sim
    \mathrm{GIG} \left(
        -\left(a_{\rho} + n + \dfrac{p_{\bm{\gamma}}}{2}\right),
        \eta \sum_{i=1}^n \dfrac{1}{\sigma^2_i},
        2b_{\rho} + \dfrac{\bm{\beta}_{\bm{\gamma}}^\top \bm{\beta}_{\bm{\gamma}}}{\tau^2} + \eta \sum_{i=1}^n \sigma^2_i
    \right)
    \ind{(\alpha = 0)} \nonumber \\
    & \quad
    + \mathrm{GIG} \left(
        \dfrac{n\eta - p_{\bm{\gamma}} - 2a_{\rho}}{2},
        \eta \sum_{i=1}^n \dfrac{1}{\sigma^2_i},
        2b_{\rho} + \dfrac{\bm{\beta}_{\bm{\gamma}}^\top \bm{\beta}_{\bm{\gamma}}}{\tau^2}
    \right)
    \ind{(\alpha = 1)},
\nonumber \\
    p(\alpha & \mid \bm{Y}, \bm{\gamma}, \bm{\beta}, \rho^2, \omega) \nonumber \\
    & \propto
    \left[
        (1-\omega) \sum_{\eta \in \mathcal{G}_{\eta}^{0}}
        \prod_{i=1}^n \dfrac{1}{2 \sqrt{\eta \rho^2} K_1(\eta)} \exp{\left(-\sqrt{\eta \left(\eta + \dfrac{\epsilon_i^2}{\rho^2}\right)}\right)}
    \right] \ind{(\alpha = 0)} \nonumber \\
    & \quad
    + \left[
        \omega \sum_{\eta \in \mathcal{G}_{\eta}^{1}}
        \prod_{i=1}^n \dfrac{\Gamma((\eta+1)/2)}{\Gamma(\eta/2)} \dfrac{\eta^{\eta/2}}{\sqrt{\pi\rho^2}} \left(\eta + \dfrac{\epsilon_i^2}{\rho^2}\right)^{-(\eta+1)/2}
    \right] \ind{(\alpha = 1)},
\nonumber \\
    p(\eta & \mid \bm{Y}, \bm{\gamma}, \bm{\beta}, \rho^2, \alpha) \nonumber \\
    & \propto
    \left[
        \left\{
            \prod_{i=1}^n \dfrac{1}{2 \sqrt{\eta \rho^2} K_1(\eta)} \exp{\left(-\sqrt{\eta \left(\eta + \dfrac{\epsilon_i^2}{\rho^2}\right)}\right)}
        \right\}
        \ind{(\eta \in \mathcal{G}_{\eta}^{0})}
    \right] \ind{(\alpha = 0)} \nonumber \\
    & \quad
    + \left[
        \left\{
            \prod_{i=1}^n \dfrac{\Gamma((\eta+1)/2)}{\Gamma(\eta/2)} \dfrac{\eta^{\eta/2}}{\sqrt{\pi\rho^2}} \left(\eta + \dfrac{\epsilon_i^2}{\rho^2}\right)^{-(\eta+1)/2}
        \right\}
        \ind{(\eta \in \mathcal{G}_{\eta}^{1})}
    \right] \ind{(\alpha = 1)},
\nonumber \\
    \sigma^2_i & \mid \bm{Y}, \bm{\gamma}, \bm{\beta}, \alpha, \eta, \rho^2 \nonumber \\
    & \stackrel{\mathrm{ind}}{\sim}
    \mathrm{GIG} \left(
        \dfrac{1}{2}, \dfrac{\eta}{\rho^2}, \epsilon_i^2 + \eta\rho^2
    \right) \ind{(\alpha = 0)}
    + \mathrm{IGamma} \left(
        \dfrac{\eta+1}{2}, \dfrac{\epsilon_i^2 + \eta\rho^2}{2}
    \right) \ind{(\alpha = 1)}, \nonumber \\
    & \qquad i = 1, 2, \dots, n,
\nonumber \\
    \tau^2 & \mid \bm{\gamma}, \bm{\beta}, \rho^2
    \sim
    \mathrm{IGamma} \left(
        \dfrac{\lambda_{\tau} + p_{\bm{\gamma}}}{2},
        \dfrac{\bm{\beta}_{\bm{\gamma}}^\top \bm{\beta}_{\bm{\gamma}}}{2\rho^2} + \dfrac{\lambda_{\tau}}{2}
    \right),
\nonumber \\
    \theta & \mid \bm{\gamma}
    \sim
    \mathrm{Beta} \left(
        c_{\theta} + p_{\bm{\gamma}},
        d_{\theta} + p - p_{\bm{\gamma}}
    \right),
\nonumber \\
    \omega & \mid \alpha
    \sim
    \mathrm{Beta} \left(
        r_{\omega} + \alpha,
        s_{\omega} + 1 - \alpha
    \right),
\label{eq:g-htem_post}
\end{align}
where $D_{\bm{\gamma}} = X_{\bm{\gamma}}^\top \Sigma^{-1} X_{\bm{\gamma}} + \dfrac{1}{\tau^2 \rho^2} I_{p_{\bm{\gamma}}}$, $p_{\bm{\gamma}} = \sum_{j=1}^p \bm{\gamma}_j$, and ${\epsilon}_i = y_i - \bm{x}^{(\bm{\gamma}) \top}_i \bm{\beta}_{\bm{\gamma}}$, for $i = 1, 2, \dots, n$, $\bm{\beta}_{\bm{\gamma}}$ is the $p_{\bm{\gamma}}$-dimensional vector of the non-zero components in $\bm{\beta}$  conditional on the inclusion indicator vector $\bm{\gamma}$, and $X_{\bm{\gamma}}$ signifies the matrix of corresponding columns of the design matrix $X$. Given a current model $\bm{\gamma}$ and a proposed model $\bm{\gamma}^*$, the Metropolis-Hastings acceptance probability is computed as
\begin{equation*}
    \pi_a = \min{\left\{
        1,
        \dfrac{p(\bm{\gamma}^* \mid \bm{Y}, \bm{\sigma}, \tau^2, \rho^2, \theta)}{p(\bm{\gamma} \mid \bm{Y}, \bm{\sigma}, \tau^2, \rho^2, \theta)}
    \right\}}.
\end{equation*}

\subsection{Choice of Hyperparameters} \label{sec:gecm_hyperpar}

The proposed GECM-HTEM algorithm involves several hyperparameters, whose choices are crucial to the performance of the algorithm. 

For the ECM step, a vital question arises regarding the choice of $\kappa_0$ and $\kappa_1$ in the prior on the regression coefficients $\bm{\beta}$ in \eqref{eq:ecm-htem_lik-priors}. Conditional on $\tau^2$, $\kappa_0$ and $\kappa_1$ control the variances of the spike and the slab components, respectively. Similarly, marginalization over $\tau^2$ leads to a mixture of multivariate Student-$t$ priors on $\bm{\beta}$ with $\lambda_{\tau}$ degrees of freedom, and scale parameters $\kappa_0\rho^2$ and $\kappa_1\rho^2$ corresponding to the spike and the slab components, respectively. In this case, $\kappa_0$ and $\kappa_1$ regulate the scales of the respective Student-$t$ distributions in the spike and the slab components. Besides, the prior specification in \eqref{eq:g-htem_lik-priors} yields a multivariate Student-$t$ prior with $\lambda_{\tau}$ degrees of freedom and scale $\rho^2$ on the non-zero regression coefficients $\bm{\beta}_{\bm{\gamma}}$ in the Gibbs sampling step, when marginalized over $\tau^2$. Since the covariates and response variables are standardized, we choose $\kappa_1 = 1$, which also preserves consistency between the priors on the slab components in the ECM and the Gibbs sampling steps. Furthermore, we take $\lambda_{\tau} = 1$ for reasonably thick-tailed priors on the slab components.
The spike hyperparameter $\kappa_0$ affects the algorithm even more critically than the slab. Rather than suggesting a single value, it is more appealing to tune $\kappa_0$ from a grid of values, based on the data. Following the grid recommended by \citet{rockova2014}, we allow $\kappa_0$ to be selected from $\{0.01, 0.02, \dots, 0.51\}$ through 10-fold cross-validation. The optimal $\kappa_0$ is chosen to be the value which minimizes the median of the fold-wise medians of absolute differences between the actual and the predicted values of the response variables. In our implementation, the cross-validation step is carried out in parallel to enhance the computational efficiency of the algorithm. Since the response variables are standardized to have a unit variance, an $\mathrm{InvGamma}(a_{\rho}=2.1, b_{\rho}=0.1)$ prior is placed on the error scale parameter $\rho^2$, such that the prior mass is mostly concentrated below one. For $\theta$, the probability of the slab component, we consider a $\mathrm{Beta}(c_{\theta}=1, d_{\theta}=1)$ prior in the Gibbs sampling step, which has also been widely used in the literature. In contrast, the ECM step uses a $\mathrm{Beta}(c_{\theta}=1.1, d_{\theta}=1.1)$ prior on $\theta$ to avoid boundary estimates.

The roles of the error distribution indicator $\alpha$ and the tail thickness parameter $\eta$ require some discussion as well. As indicated in Section \ref{sec:gecm_gecm-htem}, the ECM algorithm is run separately for each fixed choice of $\alpha\in \{0,1\}$, with $\eta$ fixed at $\eta = \eta_{\alpha}$ during the ECM step, and the preferred value of $\alpha$ is selected by comparing the corresponding maximized log-posteriors. By contrast, in the Gibbs sampler, both $\alpha$ and $\eta$ are assigned prior distributions and sampled as part of the posterior computation. We note that although no explicit prior is assigned to $\alpha$ in the ECM step, selecting the value of $\alpha$ with the larger maximized log-posterior is equivalent to using equal prior probabilities for the two error models, since such a prior contributes the same additive constant to both log-posterior criteria. This Bernoulli prior with probability 0.5 matches the marginal prior on $\alpha$ induced by the symmetric Beta-Bernoulli prior used in the Gibbs sampler. Ideally, putting a prior on $\eta$ throughout the entire algorithm would have been more appropriate for capturing tail uncertainty. However, as described towards the end of Section \ref{sec:gecm_model}, owing to the inability of suitable tail heaviness estimation by the ECM, and potential theoretical problems in estimation of the degrees of freedom in the Student-$t$ error model via EM-like algorithms \citep{fernandez1999}, the value of $\eta$ is fixed at $\eta_{\alpha}$. It might be noted that the normal and the Laplace distributions can be obtained as limiting cases of the hyperbolic distribution for large and small values of $\eta$, respectively. Therefore, for $\alpha = 0$ corresponding to hyperbolic errors, a choice of $\eta_{\alpha} = 1$ produces tails heavier than those of a normal distribution but much lighter than those of a Laplace distribution. The Student-$t$ family of distributions, on the other hand, has polynomial tails for all finite degrees of freedom $\eta$, in contrast to the exponential tails of the Laplace distribution, with the degree of tail heaviness decreasing as $\eta$ increases and the distribution converging to the normal distribution as $\eta \to \infty$. To strike a balance between the two extreme kinds of tail thickness, we choose $\eta_{\alpha} = 4.1$ while executing the ECM step for Student-$t$ errors, that is, for $\alpha = 1$. When the degree of tail heaviness is unknown and cannot be accurately estimated via the ECM step, these choices represent a practical compromise and may reduce sensitivity to model misspecification. Finally, as stated earlier in this section, in the Gibbs sampling step, we place a $\mathrm{Bernoulli}(\omega)$ prior on $\alpha$, and assign a $\mathrm{Beta}(r_{\omega}=1, s_{\omega}=1)$ prior on $\omega$ to give equal prior weights to the hyperbolic and Student-$t$ models. Given $\alpha = 0$ and $\alpha = 1$, we put discrete uniform priors on $\eta$ with prior supports $\mathcal{G}_{\eta}^{0} = \{0.05, 0.1, 0.2, 0.3, \dots, 0.9, 1, 2, 5, 10, 20, 50\}$ and $\mathcal{G}_{\eta}^{1} = \{2.1, 5, 10, 20, 50\}$, respectively, to study the unknown tail behavior more extensively.

\section{Simulation Study} \label{sec:simul}

The efficacy of the proposed GECM-HTEM technique is demonstrated through simulation studies in comparison with two state-of-the-art MAP estimation techniques, namely, EMVS \citep{rockova2014} and SSQLASSO \citep{liu2024}. Under different kinds of spike-and-slab priors, both the contenders perform deterministic model selection through posterior maximization using the EM algorithm. Specifically, the EMVS algorithm assumes a normal likelihood, while SSQLASSO incorporates a heavier-tailed asymmetric Laplace likelihood. The competing models thus enable a two-fold comparison. First, the studies help us scrutinize the potential advantages and disadvantages of our proposed hybrid approach over MAP estimators. Second, we can compare the flexibility of different likelihoods in capturing tail heaviness.

For the present analysis, SSQLASSO and EMVS are implemented using the \texttt{emBayes} \citep{liu:embayes} and \texttt{EMVS} \citep{rockova:emvs} packages in R, respectively. As recommended by \citet{rockova2014}, we used the grid $\{0.01, 0.02, \dots, 0.51\}$ for the spike hyperparameter for EMVS, which has also been used by our proposed technique. Because applying the same grid to SSQLASSO resulted in computational errors, and no default grid for the corresponding spike hyperparameter was provided in either the paper or the accompanying R package, we used the grid $\{0.001,0.002, \dots,0.051\}$, guided by Figure D3 in \citet{liu2024}. Additionally, the linear quantile regression model for SSQLASSO is run for a quantile level of 0.5 to draw analogy with our model.

To conduct the studies, we generate $n = 400$ observations from the regression model
\begin{equation}
    y_i = \beta_0 + \beta_1 x_{i1} + \beta_2 x_{i2} + \dots + \beta_p x_{ip} + \epsilon_i,
    \quad
    i = 1, 2, \dots, n,
\end{equation}
under the following scenarios:
\begin{enumerate}[label = (\Roman*)]
    \item We take $p = 1000$. The covariates $\bm{x}_i = (x_{i1}, x_{i2}, \dots, x_{ip})^\top$, $i = 1, 2, \dots, n$, are independently drawn from a multivariate normal distribution $\mathrm{N}_p(\bm{0}, \Sigma_X)$ with a first order autoregressive correlation structure as $\Sigma_X = (0.6^{|k-l|})_{1 \leq k, l \leq p}$. The true regression coefficients are $\beta_0 = 2$, $\beta_1 = \dots = \beta_{100} = 1.5$, $\beta_{101} = \dots = \beta_p = 0$. The errors are assumed to follow a $\mathrm{Hyperbolic}(\eta=0.5, \rho^2=2)$ distribution.
    \item This scenario is similar to Scenario (I), except that the errors are normally distributed with zero mean and variance 2.
    \item Here, we take $p = 1500$, and the covariates $x_{ij}$s ($i = 1, 2, \dots, n$; $j = 1, 2, \dots, p$) are independently drawn from a standard normal distribution. The true regression coefficients are $\beta_0 = 2$, $\beta_1 = \dots = \beta_{50} = 1.5$, $\beta_{51} = \dots = \beta_p = 0$, and the errors have a Student-$t$ distribution with 2.05 degrees of freedom.
\end{enumerate}
Each scenario is replicated 100 times to capture variability in the performance metrics. The studied methods are assessed in terms of estimation and out-of-sample predictive performances. For studying estimation performance, we consider the root-mean-squared errors (RMSEs) of the regression coefficients, defined as $\left[(p+1)^{-1} \sum_{j=0}^p (\beta_j - \hat{\beta}_j)^2\right]^{1/2}$, where $\hat{\beta}_j$ denotes the estimate of the $j$th regression coefficient $\beta_j$ ($j = 0, 1, 2, \dots, p$). We also compare the RMSEs of the expected values of the response variables, $\mathbb{E}(y_i | \bm{\beta}) = \sum_{j=0}^p \beta_j x_{ij}$, obtained as \(\left[n^{-1} \sum_{i=1}^n \{\sum_{j=0}^p x_{ij} (\beta_j - \hat{\beta}_j)\}^2\right]^{1/2}\), where $x_{i0} = 1$, for every $i = 1, 2, \dots, n$. Moreover, the true positive rates (TPRs) and the true negative rates (TNRs) of the methods are examined to inspect the variable selection performance. The TPR denotes the proportion of signals correctly identified by a method among all the true signals, while the TNR is the proportion of correctly dropped noise variables among all the true noise variables. Therefore, both TPR and TNR should be equal to 1 in an ideal scenario of perfect detection of signals and noise variables. Based on 1000 out-of-sample observations for each scenario, the predictive performance is evaluated using median of absolute differences (MeADs) between the actual and the predicted values of the response variables, as well as empirical coverage probabilities and median widths of 90\% prediction intervals. Lastly, we also compare the computation times of the three competing methods under consideration.

Before proceeding with the comparison, we state a few details of the implementation of the proposed GECM-HTEM algorithm. For the said method, model selection is first performed using the ECM step, followed by Gibbs sampling under the reduced model for 11000 iterations, discarding the first 1000 iterations as burnin. Moreover, for every MCMC iteration, the posterior samples of $(\beta_1, \beta_2, \dots, \beta_p)^\top$ lead to a draw of the intercept parameter $\beta_0$ as
\begin{equation}
    \beta_0 = \bar{Y} - \sum_{j=1}^p \beta_j \bar{x}_j,
\end{equation}
where $\bar{Y}$ is the mean of the non-centered and unscaled response variable, and $\bar{x}_j$ is the mean of the non-centered and unscaled $j$th covariate ($j = 1, 2, \dots, p$). Similarly, posterior samples for the expected values of the response variables, $\mathbb{E}(y_i | \bm{\beta})$s ($i = 1, 2, \dots, n$), are obtained at every MCMC iteration using the posterior samples of $(\beta_0, \beta_1, \beta_2, \dots, \beta_p)^\top$ as
\begin{equation}
    \mathbb{E}(y_i | \bm{\beta}) = \beta_0 + \sum_{j=0}^p \beta_j x_{ij}.
\end{equation}

\subsection{Estimation Performance}

At first, we evaluate the methods in terms of estimation of the regression coefficients and the expected values of the response variables. Under GECM-HTEM, the point estimates $\hat{\beta}_j$s of the regression coefficients $\beta_j$s ($j = 0, 1, \dots, p$) are taken as the component-wise medians of the posterior samples, owing to the robustness of the median over the mean. Likewise, the expected responses $\mathbb{E}(y_i | \bm{\beta})$s ($i = 1, 2, \dots, n$) are estimated using their corresponding posterior medians. On the other hand, the MAP techniques estimate the same quantities via posterior maximization. From the boxplots of the RMSEs of the regression coefficients in Figure \ref{fig:simul_gecm-map_est}, GECM-HTEM typically shows the lowest RMSEs among the three methods for the normal and the Student-$t$ error models, while exhibiting competitive performance with SSQLASSO for true hyperbolic errors. EMVS consistently produces the largest RMSEs across all the scenarios, with Scenario (I) (true hyperbolic errors) depicting the most stark difference with the two competitors. A similar observation can be made for the expected responses $\mathbb{E}(y_i | \bm{\beta})$s ($i = 1, 2, \dots, n$) from the boxplots of their RMSEs in Figure \ref{fig:simul_gecm-map_est}.
\begin{figure}[!ht]
\centering
    \includegraphics[scale=0.5]{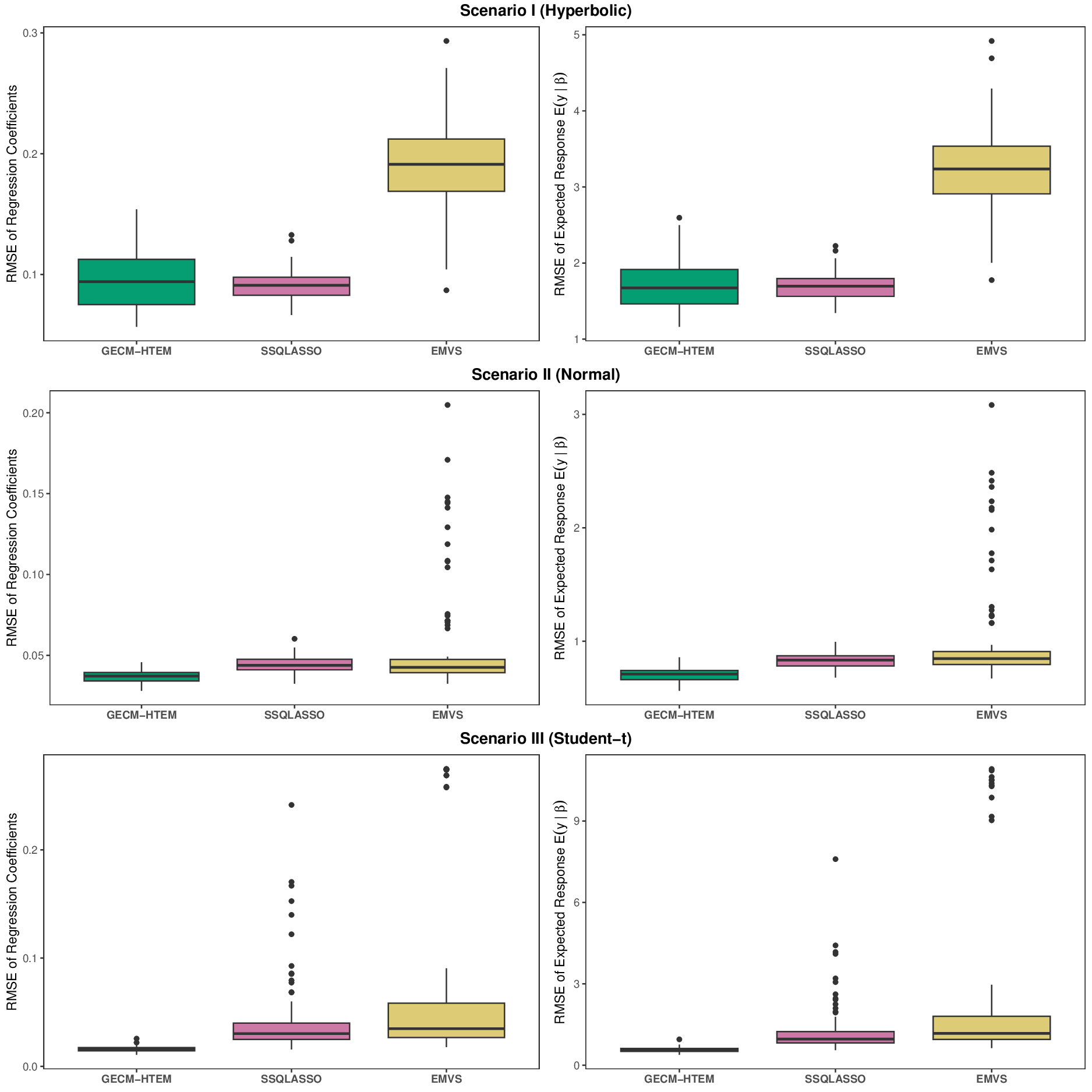}
\caption{RMSEs of all regression coefficients and the expected values of the response variables for the proposed GECM-HTEM method and the other studied MAP-based methods under different scenarios. The boxplots are based on 100 replications.}
\label{fig:simul_gecm-map_est}
\end{figure}

We now refer to the TPRs and TNRs of the methods in Table \ref{tab:simul_gecm-map_vs} to analyze their variable selection performance. For GECM-HTEM, we examine the inclusion indicators $\gamma_j$s ($j = 1, 2, \dots, p$), which are estimated using the median probability model \citep{barbieri2004}. More precisely, the $j$th inclusion indicator is estimated as $\hat{\gamma}_j = 1$, if the corresponding marginal posterior inclusion probability $\mathbb{P}(\gamma_j = 1 | \bm{Y})$ is at least 0.5. The estimated inclusion indicators are used to calculate the TPRs and TNRs against the true inclusion indicators On the other hand, for the MAP-based methods EMVS and SSQLASSO, we use the models returned by their respective R packages. Under every scenario, both the TPR and TNR for GECM-HTEM are close to 1, indicating the effectiveness of our proposed method in including signals and excluding noise covariates. SSQLASSO produces similar TPRs but lower TNRs than GECM-HTEM across all the scenarios, while EMVS shows the reverse pattern with similar TNRs and lower TPRs. In other words, SSQLASSO and EMVS favor larger and smaller models, respectively, while our method tends to make a trade-off, with a slight preference towards smaller models.
\begin{table}[!ht]
\centering
\caption{TPRs (TNRs) of all regression coefficients for the proposed GECM-HTEM method and the other studied MAP-based methods under different scenarios. The values are averaged over 100 replications.}
\begin{tabular}{ccccc}
    \hline
    & GECM-HTEM & SSQLASSO & EMVS \\
    \hline
    Scenario (I) (Hyperbolic) &  0.998 (0.998) & 1.000 (0.956) & 0.895 (1.000) \\
    Scenario (II) (Normal) &  1.000 (1.000) & 1.000 (0.707) & 0.993 (1.000) \\
    Scenario (III) (Student-$t$) &  1.000 (0.998) & 0.993 (0.947) & 0.893 (0.998) \\
    \hline
\end{tabular}
\label{tab:simul_gecm-map_vs}
\end{table}

\subsection{Predictive Performance}

The predictive performance of the three methods is analyzed and compared with a newly generated test dataset of 1000 observations. Using the posterior samples based on the training set, we obtain the point predictions for the test samples under GECM-HTEM. For uncertainty quantification, we use the estimated quantiles of the posterior predictive distribution to construct 90\% prediction intervals, and measure their empirical coverage probabilities and median widths. As far as SSQLASSO and EMVS are concerned, the point predictions are obtained using the MAP estimates based on the training samples. On the other hand, to create prediction intervals, the predictive errors are generated from their respective assumed error distributions, that is, Laplace for SSQLASSO and normal for EMVS, conditional on the estimated posterior modes.

The MeAD values of the point estimates, along with the empirical coverage probabilities and the median widths of the prediction intervals, are reported in Figure \ref{fig:simul_gecm-map_pred}. We observe that the predictive MeAD behaves similarly to the RMSEs of the regression coefficients and the mean response in Figure \ref{fig:simul_gecm-map_est}. That is, in terms of point prediction, GECM-HTEM shows competitive performance with SSQLASSO, with slight improvement for the normal and the Student-$t$ generating models. Furthermore, EMVS is consistently outperformed by yielding larger MeAD values than the other two methods. Promising behavior is also portrayed by the results on prediction intervals in Figure \ref{fig:simul_gecm-map_pred}. For Scenarios (I) and (II) corresponding to hyperbolic and normal errors, respectively, SSQLASSO has a general tendency of significant undercoverage with shorter widths compared to GECM-HTEM, while EMVS produces wider intervals than GECM-HTEM with notable overcoverage. In these scenarios, the coverage of GECM-HTEM is closer to the nominal level than that of the other methods. Under Student-$t$ errors in Scenario (III), the three methods typically exhibit similar widths and maintain a coverage close to the nominal level, although with large variability in coverage and occasional undercoverage for EMVS and SSQLASSO.
\begin{figure}[!ht]
\centering
    \includegraphics[scale=0.5]{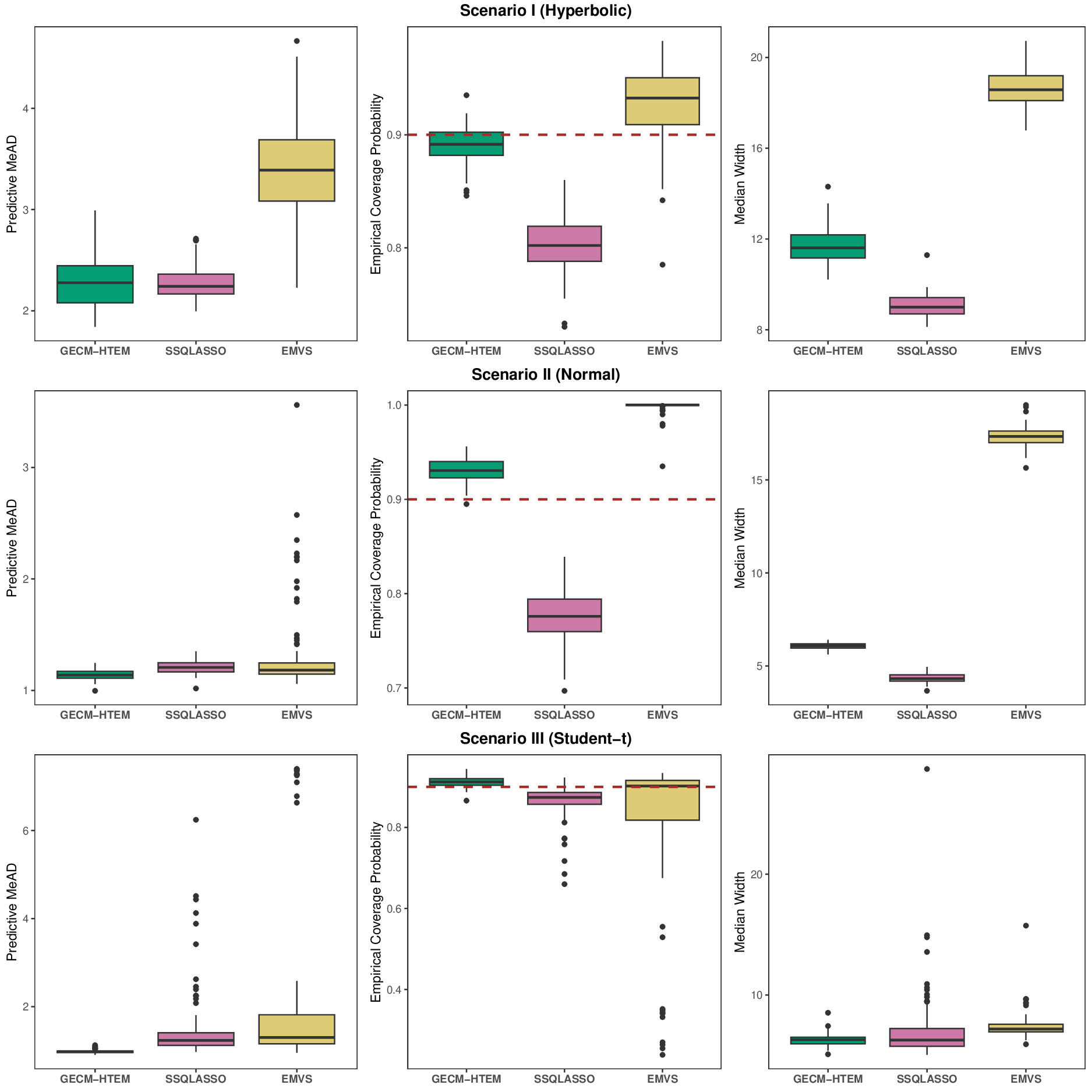}
\caption{Predictive performance of the proposed GECM-HTEM method and the other studied MAP-based methods under different scenarios. The boxplots are based on 100 replications.}
\label{fig:simul_gecm-map_pred}
\end{figure}

To summarize, GECM-HTEM exhibits competitive performance in coefficient estimation and point prediction, and superior performance in variable selection and uncertainty quantification through prediction intervals. Moreover, the consistency of the results for GECM-HTEM across the scenarios indicates its robustness and adaptability to different error distributions with varying degrees of tail heaviness.

\subsection{Computation Time} \label{sec:simul_time}

We conclude the simulation study by highlighting the computational advantage of our proposed method. Both the competing methods EMVS and SSQLASSO are implemented in C through their respective software packages, whereas our method is executed in R. Therefore, the computation times may not be strictly comparable, since lower level programming languages such as C are generally faster than higher level languages such as R. Despite this fact, we investigate the potential gain of using parallel computation to perform cross-validation for choosing the optimal value of the spike hyperparameter in the ECM step of our proposed method, which is a distinctive feature of our algorithm.

The computation times of the algorithms are reported in Table \ref{tab:simul_gecm-map_time}. For GECM-HTEM, we run the method using 1, 4, 6 and 10 CPU cores, respectively. Since any standard computer typically consists of 4 or 6 CPU cores, we choose these values to study the time taken by the proposed method in a traditional machine. The remaining choices, namely 1 and 10 cores, respectively correspond to no parallelization and parallelization over a larger number of cores than in a standard machine. We observe from Table \ref{tab:simul_gecm-map_time} that EMVS turns out to be the fastest among the three methods. On the other hand, while GECM-HTEM is initially slower than SSQLASSO using 1 core, it can be made faster through parallel computation, even though SSQLASSO is implemented in C while GECM-HTEM is implemented in R. Overall, our proposed GECM-HTEM method appears to be a computationally viable alternative for simultaneous tail estimation, variable selection and uncertainty quantification, compared to state-of-the-art MAP-based techniques.
\begin{table}[!ht]
\centering
\caption{Computation time (in minutes) for the proposed GECM-HTEM method using different number of CPU cores, and the other MAP-based methods, under a single replication of Scenario I with errors generated from the hyperbolic distribution. The running times are based on an Intel Xeon Gold 6230 processor with 80 cores.}
\begin{tabular}{cccccc}
    \hline
    \multicolumn{4}{c}{GECM-HTEM} & \multirow{2}{*}{SSQLASSO} & \multirow{2}{*}{EMVS} \\
    1 core & 4 cores & 6 cores & 10 cores & & \\
    \hline
    70.42 & 19.68 & 18.38 & 15.18 & 58.78 & 0.18 \\
    \hline
\end{tabular}
\label{tab:simul_gecm-map_time}
\end{table}

\section{Real Data Analysis} \label{sec:real}

Let us now proceed to examine the performance of GECM-HTEM using a few benchmark real datasets. In particular, we consider one genomic dataset on skin cutaneous melanoma, and two datasets from the financial sector on house prices and salaries. An additional dataset on house prices is studied in Appendix \ref{app:real_boston}, primarily to illustrate the ability of our proposed method in capturing model sparsity. Every dataset is divided into training and test sets with a 90\%-10\% split ratio, and the process is replicated 100 times to reduce sensitivity of the results to a specific choice of split. Our proposed method is compared with EMVS and SSQLASSO in terms of predictive accuracy through median absolute deviations of the predicted test sample response values, empirical coverage probabilities and median widths of 90\% prediction intervals. The results on predictive accuracy are depicted in Figure \ref{fig:real_pred}. Additionally, we inspect model parsimony across the methods for each dataset. It is worthwhile to mention that for our proposed GECM-HTEM method, the variable selection is carried out using the median probability model \citep{barbieri2004}, similar to the simulation study in Section \ref{sec:simul}.

\subsection{Skin Cutaneous Melanoma Dataset}

The Skin Cutaneous Melanoma dataset (available through the cBio Cancer Genomics Portal\footnote{\href{https://www.cbioportal.org/study/summary?id=skcm_tcga}{https://www.cbioportal.org/study/summary?id=skcm\_tcga}}) consists of multiple clinical and gene-expression features on several patients, and has been explored in genomic studies to investigate the association between gene expression profiles and cutaneous melanoma severity. After discarding the patients with missing values, the dataset contains information on 351 subjects. We consider the log-transformed Breslow's depth, a vital clinical marker of disease progression, as the response variable. The remaining clinical features, namely, age, initial cancer year and gender, are taken as covariates. It is to be noted that age and initial cancer year are continuous variables, while gender is a binary covariate. Furthermore, among the 20531 original gene-expression features, we select the top 1000 genes having the largest marginal correlations with the response variable, as covariates for our study. Accordingly, we perform our downstream analysis on 351 observations, with 1003 covariates and the logarithm of Breslow's depth as the response variable.

The results on predictive performance of the methods are reported in the top panel of Figure \ref{fig:real_pred}. We observe that while SSQLASSO typically yields the lowest predictive MeAD values, GECM-HTEM is a close competitor. On the other hand, in terms of empirical coverages and median widths of 90\% prediction intervals, both GECM-HTEM and SSQLASSO usually attain close to the nominal coverage, with slightly wider intervals for GECM-HTEM. EMVS is outperformed by the other two methods in all aspects, with larger predictive MeAD, substantial undercoverage, and overly narrow prediction intervals. In other words, GECM-HTEM exhibits satisfactory performance in uncertainty quantification through prediction intervals and a competitive performance in point prediction.

To study the model size selected by each method, we look into the clinical and gene-expression features separately. From the results in Table \ref{tab:real_skcm_vs}, it appears that for both types of features, GECM-HTEM tends to choose the smallest model, followed by EMVS and SSQLASSO. The distinction in model sizes across the methods is particularly noticeable in terms of the gene-expressions, with SSQLASSO retaining a large number of genes, as opposed to significantly smaller models for GECM-HTEM and EMVS.
\begin{table}[!ht]
\centering
\caption{Five-number summary statistics of the number of covariates (clinical and gene-expression) included in the model by the proposed GECM-HTEM method and the other studied MAP-based methods for the Skin Cutaneous Melanoma dataset. The statistics are computed based on 100 replications.}
\begin{tabular}{c|c|ccccc}
    \hline
    Covariates & Method & Minimum & 1st Quartile & Median & 3rd Quartile & Maximum \\
    \hline
    \multirow{3}{*}{\makecell{Clinical \\ (3)}}
        & GECM-HTEM & 1 & 1 & 1 & 1 & 2 \\
        & SSQLASSO & 1 & 2 & 2 & 3 & 3 \\
        & EMVS & 1 & 1 & 1 & 2 & 3 \\
    \hline
    \multirow{3}{*}{\makecell{Gene-\\Expression \\ (1000)}}
        & GECM-HTEM & 8 & 15 & 21 & 26 & 39 \\
        & SSQLASSO & 88 & 308 & 456 & 597 & 753 \\
        & EMVS & 23 & 38 & 51 & 60 & 84 \\
    \hline
\end{tabular}
\label{tab:real_skcm_vs}
\end{table}

\subsection{Ames Housing Dataset} \label{sec:real_ames}

\citet{decock2011} proposed the Ames Housing dataset (available in the \texttt{modeldata} package in R) as a touchstone for heavy-tailed regression modeling problems. The original dataset contains 2930 observations on 74 variables, including house price and 73 covariates potentially affecting the prices. To capture potential nonlinear effects, we create bins for some of the numeric covariates, thereby converting them to categorical variables. In particular, we consider 12 bins each for the year of building a house and the year of its remodeling. Thereafter, we include interaction terms between the neighborhood, and year of building, year of remodeling and year of sale (considered to be categorical with sale years as levels). Using dummy variables for each level of the categorical variables, the resulting design matrix consists of 2930 rows and 962 columns, and the log-transformed selling price is taken as the response variable.

In this example, SSQLASSO can be successfully executed for only 11 out of the 100 replicates using the \texttt{emBayes} package in R. Therefore, we report our findings for the three methods based on these 11 replicates in the main paper, and specify the results for GECM-HTEM and EMVS based on all the 100 replicates in Appendix \ref{app:real_ames}. Using the 11 replicates, the boxplots in the middle panel of Figure \ref{fig:real_pred} indicate slightly lower MeAD for SSQLASSO than GECM-HTEM. The prediction intervals for GECM-HTEM are a bit wider and nearly attain the nominal coverage, compared to a slight undercoverage by SSQLASSO. Both the methods significantly outperform EMVS, which yields large predictive MeAD, along with wide prediction intervals and gross undercoverage.

We now examine the number of predictors selected by each method based on the 11 replicates. A quick glance at the five-number summaries in Table \ref{tab:real_ames_vs} reveals a sharp distinction in the model sparsity across the methods. Both EMVS and GECM-HTEM favor smaller models, with EMVS selecting the least number of covariates. In contrast, SSQLASSO exhibits limited shrinkage by dropping only a few covariates in most replicates.
\begin{table}[!ht]
\centering
\caption{Five-number summary statistics of the number of covariates included in the model by the proposed GECM-HTEM method and the other studied MAP-based methods for the Ames Housing dataset. The statistics are computed based on 11 replications for which SSQLASSO can be successfully executed.}
\begin{tabular}{c|ccccc}
    \hline
    Method & Minimum & 1st Quartile & Median & 3rd Quartile & Maximum \\
    \hline
    GECM-HTEM & 43 & 48 & 54 & 56 & 60 \\
    SSQLASSO & 827 & 922 & 929 & 947 & 955 \\
    EMVS & 7 & 10 & 14 & 15 & 15 \\
    \hline
\end{tabular}
\label{tab:real_ames_vs}
\end{table}

\subsection{University of Vermont Faculty Salaries Dataset} \label{sec:real_salary}
We analyze data on faculty salaries for the year 2020 at the University of Vermont obtained from Kaggle, using the log-transformed salary as the response variable, with job title, department and college serving as predictors. Furthermore, since all the 3 covariates are categorical with multiple levels, we create dummy variables for each level, and also incorporate the interactions between colleges and job titles. The resulting dataset consists of 1756 observations, with 936 covariates and the logarithm of salary as the response variable.

It is to be noted that 605 of the interaction terms are not observed in the data. Accordingly, these 605 columns in the design matrix contain only zero entries across all rows. Instead of removing these columns manually, we retain them to check if the studied methods can effectively discard them. We observe that while GECM-HTEM and EMVS can successfully exclude these columns, the implementation of SSQLASSO via the \texttt{emBayes} package results in a runtime error. Therefore, we execute SSQLASSO on the smaller design matrix containing the remaining 331 columns, after manually dropping the 605 columns with all zero entries.

The bottom panel of Figure \ref{fig:real_pred} shows the predictive performance of the methods. The predictive MeADs for GECM-HTEM and SSQLASSO are close to one another, with slightly lower values for GECM-HTEM, while those for EMVS are substantially larger. All the methods exhibit undercoverage for 90\% prediction intervals. The coverage is largest for EMVS with substantially large widths. The proposed GECM-HTEM method attains a coverage a bit lower than that of EMVS, along with a smaller width. SSQLASSO produces narrower intervals and notably lower coverage than the other methods.

To compare the model complexities, we inspect the number of predictors selected by each method among the 331 covariates used in fitting SSQLASSO. The results in Table \ref{tab:real_salary_vs} repeat a similar story like the Ames Housing dataset. EMVS chooses the smallest model, followed by GECM-HTEM, while SSQLASSO retains a substantial number of predictors.
\begin{table}[!ht]
\centering
\caption{Five-number summary statistics of the number of covariates (among the 331 covariates used to run SSQLASSO) included in the model by the proposed GECM-HTEM method and the other studied MAP-based methods for the Faculty Salaries dataset. The statistics are computed based on 100 replications.}
\begin{tabular}{c|ccccc}
    \hline
    Method & Minimum & 1st Quartile & Median & 3rd Quartile & Maximum \\
    \hline
    GECM-HTEM & 28 & 43 & 44 & 47 & 53 \\
    SSQLASSO & 87 & 319 & 323 & 325 & 328 \\
    EMVS & 5 & 5 & 5 & 5 & 5 \\
    \hline
\end{tabular}
\label{tab:real_salary_vs}
\end{table}

\begin{figure}[!ht]
\centering
    \includegraphics[scale=0.5]{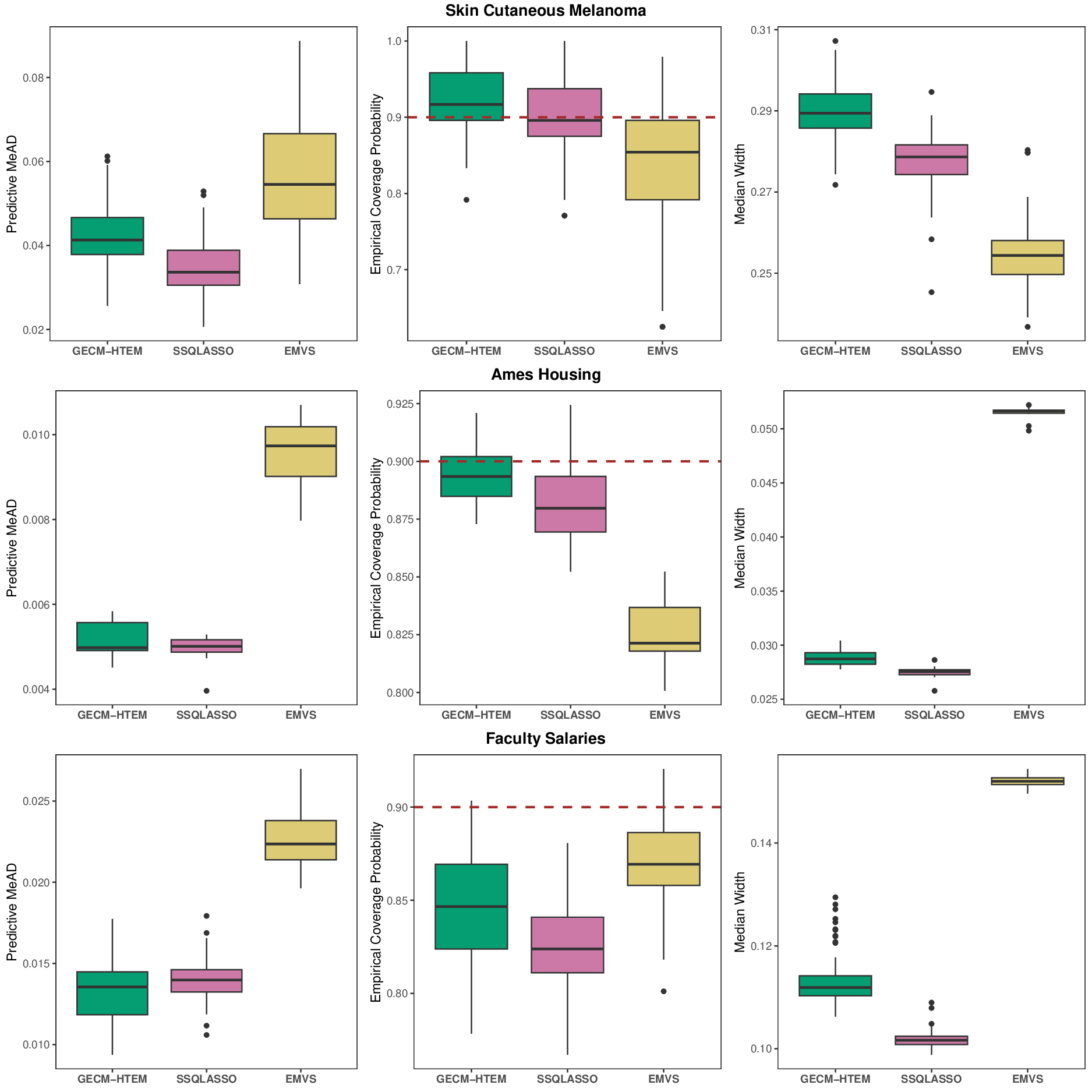}
\caption{Predictive performance of the proposed GECM-HTEM method and the other studied MAP-based methods for different real datasets. In the plots for empirical coverage probability, the dashed lines mark a coverage of 90\%. The boxplots for Skin Cutaneous Melanoma and Faculty Salaries datasets are based on 100 replications. The boxplots for Ames Housing dataset are based on 11 replications for which SSQLASSO can be successfully executed.}
\label{fig:real_pred}
\end{figure}

\section{Conclusion} \label{sec:conc}

In this paper, we have proposed a two-step approach, namely, GECM-HTEM, by combining MAP estimation through the ECM algorithm, along with a Gibbs sampler, for performing model selection and posterior computation within a flexible-tailed error framework. Our model accommodates tail behavior ranging from the light tails of the normal distribution to the heavier tails of the hyperbolic and Student-$t$ distributions, while estimating tail heaviness from the data. This approach leverages the computational advantage of MAP estimation and inferential benefits of MCMC-based algorithms. Compared with two representative MAP estimators, the proposed GECM-HTEM method has shown competitive performance in coefficient estimation and point prediction, and has demonstrated superior performance in variable selection and uncertainty quantification.

Several other possibilities can be investigated as a sequel to this paper. Studying the Bayesian bootstrap \citep{rubin1981, nie2023bbssl, nie2023deepbb} as an alternative to the Gibbs sampler for uncertainty quantification is a viable direction. Another possible option is to explore global-local shrinkage priors, like the horseshoe prior \citep{carvalho2009, carvalho2010, de2024, fan2026}, in place of spike-and-slab priors on the regression coefficients. Lastly, inducing asymmetry in the error distribution serves as a suitable extension for further work.

\section*{Data Availability}

Data sharing is not applicable to this article as no new data were created or analyzed in this study.

\section*{Use of Generative Artificial Intelligence (AI)}

Generative AI tools were used to check algebraic derivations, code and grammar. All content was subsequently reviewed and edited by the authors, who assume responsibility for the final manuscript.

\section*{Acknowledgments}

A part of this work was completed by Shamriddha De in his doctoral thesis under the supervision of Joyee Ghosh.

\bibliographystyle{apalike}
\bibliography{References}

\begin{appendices}

\section{Proof of Proposition \ref{prop:scalemix_hyperbolic}} \label{app:scalemix_hyperbolic}

\textbf{Proposition 1.} \textit{Let $A$ and $a$ be random variables such that $A|a^2 \sim \mathrm{N}(0, \rho^2a^2)$ and $a^2 \sim \mathrm{GIG}(1, \eta, \eta)$. Then $A$ is distributed as $\mathrm{Hyperbolic}(\eta, \rho^2)$.}

\begin{proof}
The conditional density of $A$ given $a^2$ is
\begin{equation*}
    f_A(A \mid a^2)
    = \dfrac{1}{\sqrt{2\pi\rho^2a^2}} \exp{\left(-\dfrac{A^2}{2\rho^2a^2}\right)},
    \quad -\infty < A < \infty,
\end{equation*}
while the density of $a^2$ is
\begin{equation*}
    f_{a^2}(a^2)
    = \dfrac{1}{2 K_1(\eta)} \exp{\left\{-\dfrac{\eta}{2} \left(a^2 + \dfrac{1}{a^2}\right)\right\}},
    \quad x > 0.
\end{equation*}
The density of $A$, marginalized over $a^2$, is obtained as
\begin{align*}
    f_A(A)
    & = \int_0^\infty f_A(A \mid a^2) f_{a^2}(a^2) da^2, \\
    & = \dfrac{1}{2 K_1(\eta) \sqrt{2\pi\rho^2}} \int_0^\infty \dfrac{1}{\sqrt{a^2}} \exp{\left\{-\dfrac{1}{2} \left( \eta a^2 + \dfrac{\eta + A^2/\rho^2}{a^2} \right)\right\}} da^2.
\end{align*}
The above integral can be evaluated using the $\mathrm{GIG}(\lambda, a, b)$ kernel in \eqref{eq:gigdensity} with $\lambda = 1/2$, $a = \eta$, $b = \eta + A^2/\rho^2$, thereby yielding
\begin{equation*}
    f_A(A)
    = \dfrac{1}{\sqrt{2\pi\rho^2} K_1(\eta)} \left(\dfrac{\eta + A^2/\rho^2}{\eta}\right)^{1/4} K_{1/2}\left(\sqrt{\eta(\eta + A^2/\rho^2)}\right),
\end{equation*}
which further simplifies to the density function of the $\mathrm{Hyperbolic}(\eta, \rho^2)$ distribution in \eqref{eq:hyperbolicdensity}, using the closed form expression $K_{1/2}(x) = \sqrt{\pi/2x} \, e^{-x}$.
\end{proof}

\section{Proof of Proposition \ref{prop:ecm-htem_max}} \label{app:ecm-htem_max}

\textbf{Proposition 2.} \textit{Under the likelihood and prior specification in \eqref{eq:ecm-htem_lik-priors}, let us assume that $\kappa_0 > 0$, $\kappa_1 > 0$, $a_{\rho} > 0$, $b_{\rho} > 0$, $\lambda_{\tau} > 0$, $c_\theta > 1$ and $d_\theta > 1$. For each fixed error model, namely the hyperbolic error model with $\alpha = 0$ and $\eta = \eta_0$, and the Student-$t$ error model with $\alpha = 1$ and $\eta = \eta_1$, conditional on the quantities computed in the E-step and on the current values of the remaining parameters, each continuous CM-step with respect to $\bm{\beta}$, $\rho^2$, $\tau^2$, $\theta$, $\sigma^2_i$ ($i = 1, 2, \dots, n$), has a unique global maximizer in the interior of the corresponding parameter space. Moreover, the unique maximizer in each case is attained by the corresponding closed-form update in \eqref{eq:ecm-htem_cmstep}.}

\begin{proof}
Throughout the proof, we assume that the error model is fixed either at the hyperbolic error model $(\alpha = 0, \eta = \eta_0)$ or at the Student-$t$ error model $(\alpha = 1, \eta = \eta_1)$, with their respective fixed shape parameters $\eta = \eta_0$ and $\eta = \eta_1$. Furthermore, while performing maximization over a parameter, we condition on the quantities computed in the E-step as well as the current values of the remaining parameters. From (\ref{eq:ecm-htem_estep1}), we recall that 
\begin{equation*}
    \mathbb{E}^*_{\bm{\gamma}}[\gamma_j]
    = \mathbb{P}(\gamma_j = 1 \mid \bm{Y}, \alpha, \eta = \eta_{\alpha}, \hat{\bm{\beta}}, \hat{\rho}^2, \hat{\theta}, \hat{\tau}^2, \hat{\bm{\sigma}})
    = g_j
\end{equation*}
which implies $0 \leq g_j \leq 1$. From \eqref{eq:ecm-htem_objective} and \eqref{eq:ecm-htem_step2}, we have 
$
    \Psi
    = \diag((\mathbb{E}^*_{\bm{\gamma}}[1/v_j]))
    = \diag\left( \frac{g_j}{\kappa_1} + \frac{1-g_j}{\kappa_0} \right)
$, 
for $j = 1, 2, \dots, p$. Since $\kappa_0 > 0$ and $\kappa_1 > 0$, every diagonal entry of $\Psi$ is positive, and accordingly, $\Psi$ is a positive definite matrix. Also, since $\sigma_i^2$, $i = 1, \dots, n$, the matrix $\Sigma^{-1}$ is positive definite.
\begin{enumerate}

\item 
We first consider the CM-step with respect to $\bm{\beta}$. Let $q_{\bm{\beta}}(\bm{\beta})$ denote the terms in $Q^{\alpha}$ that depend on $\bm{\beta}$, with the E-step quantities and all remaining parameters held fixed. Then we have
\begin{equation*}
    q_{\bm{\beta}}(\bm{\beta})
    =
    - \frac{1}{2\rho^2} (\bm{Y} - X\bm{\beta})^\top \Sigma^{-1} (\bm{Y} - X\bm{\beta})
    - \frac{1}{2\rho^2\tau^2} \bm{\beta}^\top \Psi \bm{\beta}.
\end{equation*}
The Hessian of $q_{\bm{\beta}}(\bm{\beta})$ is
\begin{equation*}
    \nabla_{\bm{\beta}}^2 q_{\bm{\beta}}(\bm{\beta})
    =
    -\frac{1}{\rho^2} \left( X^\top\Sigma^{-1}X + \frac{1}{\tau^2}\Psi \right).
\end{equation*}
For any $\bm{h} \in \mathbb{R}^p$ such that $\bm{h} \neq \bm{0}$,
\begin{equation*}
    \bm{h}^\top \left( X^\top\Sigma^{-1}X + \frac{1}{\tau^2}\Psi \right) \bm{h}
    =
    (X\bm{h})^\top\Sigma^{-1}(X\bm{h}) + \frac{1}{\tau^2}\bm{h}^\top\Psi\bm{h}
    > 0.
\end{equation*}
This is because, the first term is nonnegative since $\Sigma^{-1}$ is positive definite, while the second term is positive since $\tau^2 > 0$, $\Psi$ is positive definite and $\bm{h} \neq \bm{0}$. Hence the sum is positive. Thus 
$
    \left( X^\top\Sigma^{-1}X + \frac{1}{\tau^2}\Psi \right)
$
is a positive definite matrix, which implies 
$
    \nabla_{\bm{\beta}}^2 q_{\bm{\beta}}(\bm{\beta})
    = -\frac{1}{\rho^2} \left( X^\top\Sigma^{-1}X + \frac{1}{\tau^2}\Psi \right)
$ 
is a negative definite matrix. Therefore, $q_{\bm{\beta}}(\bm{\beta})$ is strictly concave in $\bm{\beta}$. Setting the gradient of $q_{\bm{\beta}}(\bm{\beta}) $ to $\bm{0}$, that is, $\nabla_{\bm{\beta}}q_{\bm{\beta}}(\bm{\beta}) = \bm{0}$, we have
\begin{equation*}
    X^\top\Sigma^{-1}\bm{Y} - X^\top\Sigma^{-1}X\bm{\beta} - \frac{1}{\tau^2}\Psi\bm{\beta}
    =
    \bm{0},
\end{equation*}
or equivalently
\begin{equation*}
    \left( X^\top\Sigma^{-1}X + \frac{1}{\tau^2}\Psi \right) \bm{\beta}
    =
    X^\top \Sigma^{-1} \bm{Y}.
\end{equation*}
Since 
$
    \left( X^\top\Sigma^{-1}X + \frac{1}{\tau^2}\Psi \right)
$
is positive definite, the above equation has a unique solution. By strict concavity of $q_{\bm{\beta}}(\bm{\beta})$, that solution is the unique global maximizer, and is the closed-form $\bm{\beta}$-update in \eqref{eq:ecm-htem_cmstep}.

\item 
Next, we consider the CM-step with respect to $\rho^2$. Let $z = \rho^2 > 0$, $S(\bm{\beta}, \bm{\sigma}) = (\bm{Y} - X\bm{\beta})^\top \Sigma^{-1} (\bm{Y} - X\bm{\beta})$, and $B(\bm{\beta}) = \bm{\beta}^\top \Psi \bm{\beta}$. Let $q_{z}(z)$ denote the terms in $Q^{\alpha}$ that depend on $z$ with the E-step quantities and all remaining parameters held fixed. Then we have,
\begin{equation*}
    q_{z}(z)
    =
    - \left( \frac{n+p}{2} + a_{\rho} + 1 \right) \log{z}
    - \frac{1}{z} \left\{ \frac{S(\bm{\beta},\bm{\sigma})}{2} + \frac{B(\bm{\beta})}{2\tau^2} + b_{\rho} \right\}.
\end{equation*}
Let 
$
    A = \frac{n+p}{2} + a_{\rho} + 1
$ 
and 
$
    C = \frac{S(\bm{\beta},\bm{\sigma})}{2} + \frac{B(\bm{\beta})}{2\tau^2} + b_{\rho}
$. 
Then we have $A > 0$ and $C > 0$, and 
$
    q_{z}(z) = -A\log{z} - \frac{C}{z}
$. 
Accordingly, the first derivative of $q_z(z)$ is given by
$
    q_{z}'(z) = -\frac{A}{z} + \frac{C}{z^2} = \frac{C-A z}{z^2}
$. 
Since $z^2 > 0$, we have $q_{z}'(z) > 0$ for $z < C/A$ and $q_{z}'(z) < 0$ for $z > C/A$. Therefore, $q_z(z)$ increases up to $C/A$ and decreases thereafter. Thus the unique global maximizer is attained at $z = C/A > 0$, which is exactly the closed-form $\rho^2$-update in \eqref{eq:ecm-htem_cmstep}.

\item 
The proof for the CM-step with respect to $\tau^2$ is analogous to that of $\rho^2$. Let $u = \tau^2 > 0$. Let $q_u(u)$ denote the part of $Q^{\alpha}$ that depends on $u$, while the E-step quantities and other parameters are held fixed. We have 
\begin{equation*}
    q_{u}(u)
    =
    - \left( \frac{p + \lambda_{\tau}}{2} + 1 \right) \log{u}
    - \frac{1}{u} \left\{ \frac{\bm{\beta}^\top\Psi\bm{\beta}}{2\rho^2} + \frac{\lambda_{\tau}}{2} \right\}.
\end{equation*}
Let 
$
    D = \frac{p + \lambda_{\tau}}{2} + 1
$ 
and 
$
    E = \frac{\bm{\beta}^\top\Psi\bm{\beta}}{2\rho^2} + \frac{\lambda_\tau}{2}
$. 
Then we have $D > 0$ and $E > 0$, and 
$
    q_u(u) = -D\log{u} - \frac{E}{u}
$.
The first derivative of $q_u(u)$ is
$
    q_u'(u) = -\frac{D}{u} + \frac{E}{u^2} = \frac{E - Du}{u^2}.
$
Since $u^2 > 0$, the derivative is positive for $u < E/D$ and
negative for $u>E/D$, which implies $q_u(u)$ increases up to $E/D$ and decreases thereafter. Thus the unique global maximizer is attained at $u = E/D > 0$, which is exactly the closed-form $\tau^2$-update in \eqref{eq:ecm-htem_cmstep}.

\item 
We now consider the CM-step with respect to $\theta$. Let $q_{\theta}(\theta)$ denote the part of $Q^{\alpha}$ that depends on $\theta$, while the E-step quantities and other parameters are held fixed.
Then we have
\begin{equation*}
    q_{\theta}(\theta)
    =
    \left( c_{\theta} + \sum_{j=1}^p g_j - 1 \right) \log{\theta}
    +
    \left( d_{\theta} + \sum_{j=1}^p (1-g_j) - 1 \right) \log{(1-\theta)},
\end{equation*}
where $0 < \theta < 1$. Let
$
    F = c_{\theta} + \sum_{j=1}^p g_j - 1
$ 
and
$
    G = d_{\theta} + \sum_{j=1}^p (1-g_j) - 1
$.
Because $c_{\theta} > 1$, $d_{\theta} > 1$ and $0 \leq g_j \leq 1$, we have $F > 0$ and $G > 0$. It follows that 
$
    q_{\theta}(\theta) = F\log{\theta} + G\log{(1-\theta)}
$ 
and its derivative is
$
    q_{\theta}'(\theta) = \frac{F}{\theta} - \frac{G}{1-\theta}
$. 
So, the second derivative becomes
$
    q_{\theta}''(\theta) = -\frac{F}{\theta^2} - \frac{G}{(1-\theta)^2}
    < 0
$, 
for $0 < \theta < 1$. Thus $q_{\theta}(\theta)$ is strictly concave on $(0, 1)$. Solving $q_{\theta}'(\theta) = 0$ gives $\theta = F / (F+G)$, which belongs to $(0, 1)$. Hence this point is the unique global maximizer in the interior of the parameter space of $\theta$, and it is exactly the closed-form $\theta$-update in \eqref{eq:ecm-htem_cmstep}.

\item 
Lastly, it remains to verify the CM-steps corresponding to $\sigma_i^2$, $i = 1, 2, \dots, n$. We fix $i \in \{1, \dots, n\}$, and let $s_i = \sigma_i^2 > 0$ and $e_i = y_i - \bm{x}_i^\top\bm{\beta}$, where $\bm{x}_i$ is the $i$th row of the design matrix $X$. Let $q_{s_i, \alpha}(s_i)$ denote the part of $Q^{\alpha}$ that depends on $s_i$, while the E-step quantities and other parameters are held fixed. We consider the two cases $\alpha = 0$ and $\alpha = 1$ separately.
\begin{enumerate}

\item At first, we first consider the hyperbolic error model with $\alpha = 0$ and $\eta = \eta_0$. In this case,  we have 
\begin{equation*}
    q_{s_i, 0}(s_i)
    =
    - \frac{1}{2}\log{s_i}
    - \frac{e_i^2}{2 \rho^2 s_i}
    - \frac{\eta_0}{2} \left( s_i + \frac{1}{s_i} \right).
\end{equation*}
Differentiating $q_{s_i, 0}(s_i)$ with respect to $s_i$ gives
\begin{equation*}
    q_{s_i, 0}'(s_i)
    = - \frac{1}{2s_i} + \frac{e_i^2}{2\rho^2s_i^2} - \frac{\eta_0}{2} + \frac{\eta_0}{2s_i^2}
    = \frac{1}{2s_i^2} \left\{ - \eta_0s_i^2 - s_i + \eta_0 + \frac{e_i^2}{\rho^2} \right\}
    = \frac{h_{s_i, 0}(s_i)}{2s_i^2},
\end{equation*}
where $h_{s_i, 0}(s_i) = -\eta_0s_i^2 - s_i + \eta_0 + \frac{e_i^2}{\rho^2}$. Since $2s_i^2 > 0$, the sign of $q_{s_i, 0}'(s_i)$ is same as the sign of $h_{s_i, 0}(s_i)$. We note that solving $h_{s_i, 0}(s_i) = 0$ is equivalent to solving
$
    \eta_0s_i^2 + s_i - \left(\eta_0 + \frac{e_i^2}{\rho^2}\right) = 0
$. 
By the quadratic formula, the positive and negative roots are
\begin{equation*}
    s_{i, 0}^{+}
    = \frac{
        -1
        + \sqrt{
            1 + 4\eta_0 + \left( \eta_0+\frac{e_i^2}{\rho^2} \right)
        }
    }{2\eta_0}
    \text{ and }
    s_{i, 0}^{-}
    = \frac{
        -1
        - \sqrt{
            1 + 4\eta_0 + \left( \eta_0+\frac{e_i^2}{\rho^2} \right)
        }
    }{2\eta_0},
\end{equation*}
respectively. Clearly, the negative root $s_{i, 0}^{-}$ lies outside the domain
$s_i > 0$. Now, we have 
$
    h_{s_i, 0}(s_i) = -\eta_{0} (s_i - s_{i, 0}^{+}) (s_i - s_{i, 0}^{-})
$. 
As $\eta_0 > 0$ and $(s_i - s_{i, 0}^{-}) > 0$, the sign of $h_{s_i, 0}(s_i)$ is determined by $-(s_i - s_{i, 0}^{+})$. Thus $h_{s_i, 0}(s_i) > 0$ if $0 < s_i < s_{i, 0}^{+}$ and $h_{s_i, 0}(s_i) < 0$ if $s_{i, 0}^{+} < s_i < \infty$. Since the sign of $q_{s_i, 0}'(s_i)$ is same as that of $h_{s_i, 0}(s_i)$, it follows that $q_{s_i, 0}(s_i)$ increases on $(0, s_{i, 0}^{+})$ and decreases on $(s_{i, 0}^{+}, \infty)$. Therefore, $s_{i, 0}^{+}$ is the unique global maximizer, which matches with  the $\sigma_i^2$-update for the hyperbolic error model in \eqref{eq:ecm-htem_cmstep}.

\item For the Student-$t$ error model with $\alpha = 1$ and $\eta = \eta_1$, we have
\begin{align*}
    q_{s_i, 1}(s_i)
    & = - \frac{1}{2} \log{s_i}
        - \frac{e_i^2}{2 \rho^2 s_i}
        - \left(\frac{\eta_1}{2} + 1\right) \log{s_i}
        - \frac{\eta_1}{2s_i} 
\\
    & = - \frac{\eta_1 + 3}{2} \log{s_i}
        - \frac{1}{2s_i} \left(\eta_1 + \frac{e_i^2}{\rho^2}\right).
\end{align*}
Differentiating $q_{s_i, 1}(s_i)$ with respect to $s_i$ gives
\begin{equation*}
    q_{s_i, 1}'(s_i)
    = - \frac{\eta_1+3}{2s_i}
        + \frac{1}{2s_i^2} \left( \eta_1 + \frac{e_i^2}{\rho^2} \right)
    = \frac{1}{2s_i^2} \left\{
        \eta_1
        + \frac{e_i^2}{\rho^2}
        - (\eta_1 + 3)s_i
    \right\}
\end{equation*}
Since $2s_i^2 > 0$, 
$q_{s_i, 1}'(s_i) > 0$ 
for $0 < s_i < \frac{\eta_1 + \frac{e_i^2}{\rho^2}}{\eta_1 + 3}$ 
and 
$q_{s_i, 1}'(s_i) < 0$ 
for $s_i > \frac{\eta_1 + \frac{e_i^2}{\rho^2}}{\eta_1 + 3}$. 
Thus $q_{s_i,1}(s_i)$ increases up to 
$\frac{\eta_1 + \frac{e_i^2}{\rho^2}}{\eta_1 + 3}$ 
and decreases thereafter. Thus the unique global maximizer is attained at 
$
    s_{i, 1}
    = \frac{\eta_1 + \frac{e_i^2}{\rho^2}}{\eta_1 + 3}
    > 0
$, 
since $\eta_1 > 0$ and $e_i^2/\rho^2 \geq 0$. This matches with the $\sigma_i^2$-update in \eqref{eq:ecm-htem_cmstep} for the Student-$t$ error model.

\end{enumerate}

\end{enumerate}
Combining the preceding arguments, each continuous CM-step with respect to $\bm{\beta}$, $\rho^2$, $\tau^2$, $\theta$, $\sigma^2_i$ ($i = 1, 2, \dots, n$) has a unique global maximizer in the interior of the corresponding parameter space, and the unique maximizer is attained by the corresponding closed-form update in \eqref{eq:ecm-htem_cmstep}.
\end{proof}

\section{Additional Results for Real Data Analysis} \label{app:real}

\subsection{Ames Housing Dataset} \label{app:real_ames}

In this section, we reproduce the results for GECM-HTEM and EMVS based on all the 100 replicates intended for our analysis. The results on both predictive performance (Figure \ref{fig:app_real_ames_pred}) and model size (Table \ref{tab:app_real_ames_vs}) for the methods are nearly similar to those stated in Section \ref{sec:real_ames}, with GECM-HTEM visibly outperforming EMVS in every aspect of comparison.
\begin{figure}[!ht]
\centering
    \includegraphics[scale=0.6]{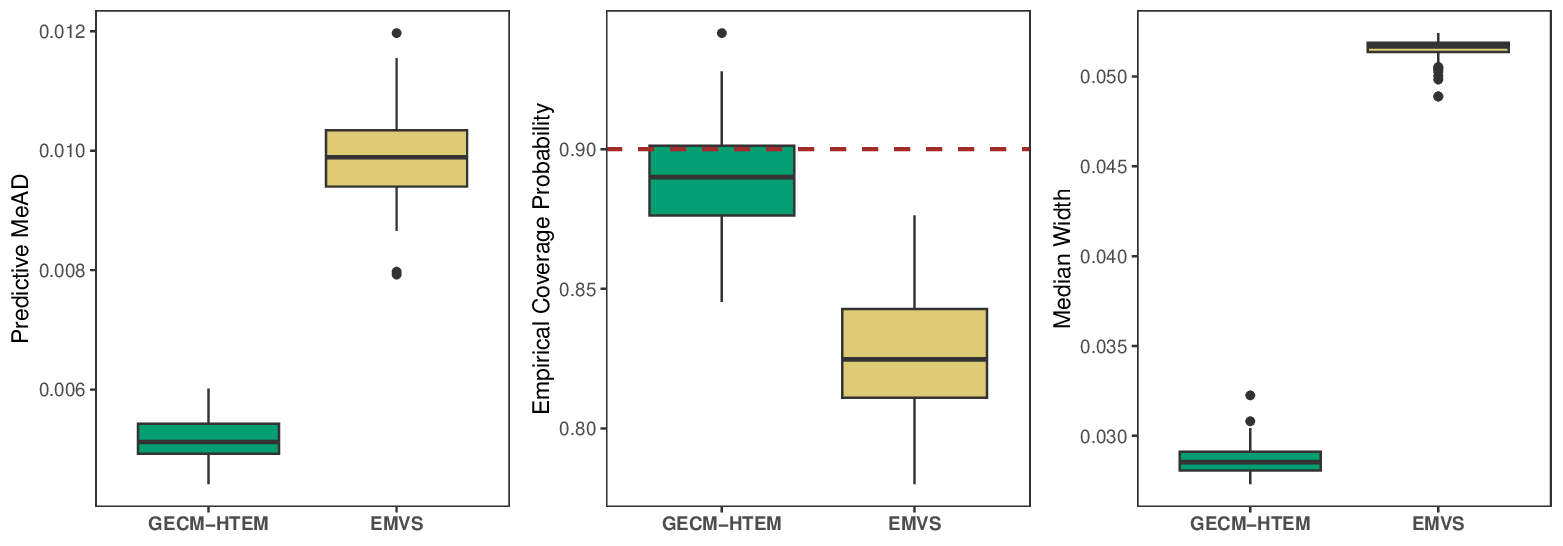}
\caption{Predictive performance of the proposed GECM-HTEM method and EMVS for Ames Housing dataset. In the middle plot, the dashed line marks a coverage of 90\%. The boxplots are based on 100 replications.}
\label{fig:app_real_ames_pred}
\end{figure}
\begin{table}[!ht]
\centering
\caption{Five-number summary statistics of the number of covariates included in the model by the proposed GECM-HTEM method and EMVS for the Ames Housing dataset. The statistics are computed based on 100 replications.}
\begin{tabular}{c|ccccc}
    \hline
    Method & Minimum & 1st Quartile & Median & 3rd Quartile & Maximum \\
    \hline
    GECM-HTEM & 40 & 49 & 54 & 55 & 62 \\
    EMVS & 3 & 13 & 14 & 14 & 16 \\
    \hline
\end{tabular}
\label{tab:app_real_ames_vs}
\end{table}

\subsection{Boston Housing Dataset} \label{app:real_boston}

The Boston Housing dataset from the \texttt{MASS} package in R, which has been extensively used as a benchmark example for regression modeling with heavy-tailed errors, is used as an illustrative example in this paper. It consists of 506 observations on median house value, regressed on 13 neighborhood characteristics as covariates. To achieve an approximate symmetry in the distribution of the residuals, the logarithm of the median house value is considered as the response variable. Following \citet{ghosh2013, de2024, de2026}, we add 1000 noise variables, independently generated from a standard normal distribution, resulting in $p = 1013$, to make the problem more interesting for variable selection. This enables us to compare how effectively the competing methods can drop the noise variables from the model, which can in turn affect the predictive performance.

From the boxplots in Figure \ref{fig:app_real_boston_pred}, we observe that while SSQLASSO typically yields the lowest predictive MeAD values, GECM-HTEM is a close competitor. On the other hand, in terms of empirical coverage and median width of 90\% prediction intervals, both GECM-HTEM and SSQLASSO usually attain coverage close to the nominal level, with slightly wider intervals for GECM-HTEM. EMVS is outperformed by the other two methods in all aspects, with larger predictive MeAD, substantial undercoverage, and overly narrow prediction intervals. In other words, GECM-HTEM exhibits satisfactory performance in uncertainty quantification through prediction intervals and a competitive performance in point prediction.
\begin{figure}[!ht]
\centering
    \includegraphics[scale=0.6]{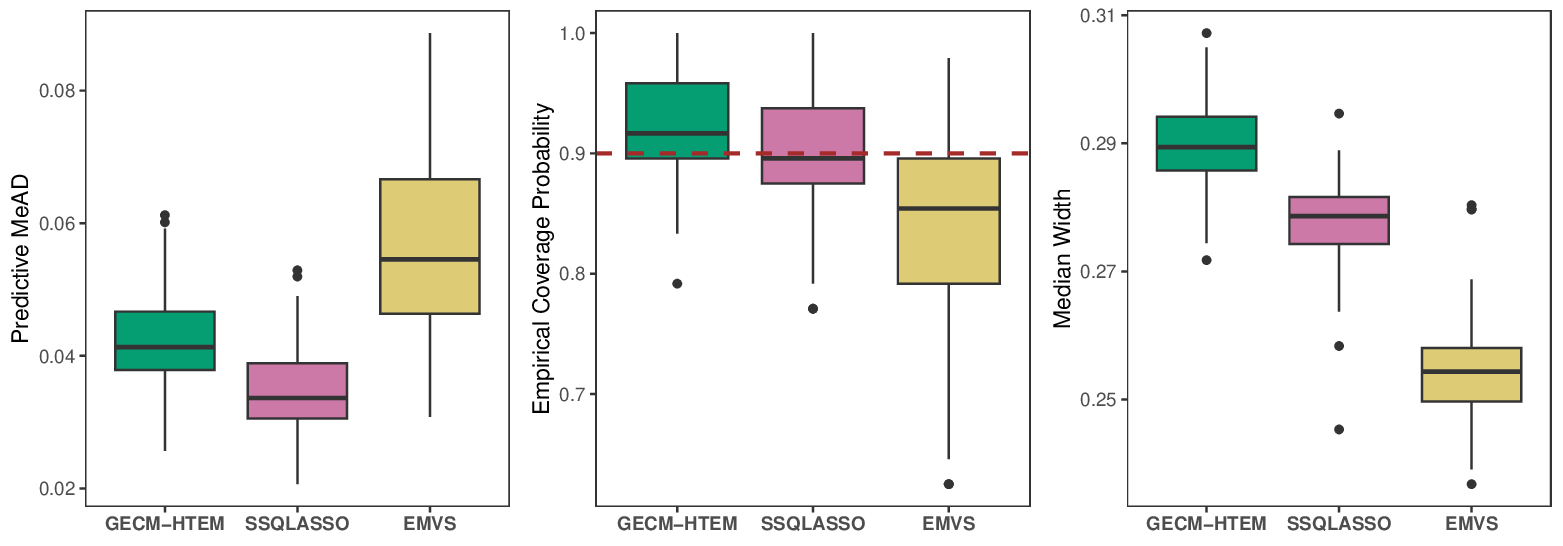}
\caption{Predictive performance of the proposed GECM-HTEM method and the other studied MAP-based methods for Boston Housing dataset. In the middle plot, the dashed line marks a coverage of 90\%. The boxplots are based on 100 replications.}
\label{fig:app_real_boston_pred}
\end{figure}

We now investigate the model size chosen by each method. The design matrix is known to contain 1000 noise variables, along with the 13 original predictors. With this information up front, it will be useful to study these two types of covariates separately. In particular, for each type, we look at the five-number summary statistics of the number of covariates included in the model. The summaries are taken over the 100 replicates and are reported in Table \ref{tab:app_real_boston_vs}. It is evident from Table \ref{tab:app_real_boston_vs} that for most of the replicates, all the three methods perform similarly in choosing the original covariates and dropping the noise variables. However, we also observe that SSQLASSO includes nearly all of the original covariates and retains a large number of noise predictors in some replicates, indicating its tendency to choose bigger models.
\begin{table}[!ht]
\centering
\caption{Five-number summary statistics of the number of covariates (original and noise) included in the model by the proposed GECM-HTEM method and the other studied MAP-based methods for the Boston Housing dataset. The statistics are computed based on 100 replications.}
\begin{tabular}{c|c|ccccc}
    \hline
    Covariates & Method & Minimum & 1st Quartile & Median & 3rd Quartile & Maximum \\
    \hline
    \multirow{3}{*}{\makecell{Original \\ (13)}}
        & GECM-HTEM & 4 & 5 & 5 & 5 & 5 \\
        & SSQLASSO & 3 & 4 & 4 & 5 & 12 \\
        & EMVS & 3 & 3 & 3 & 3 & 5 \\
    \hline
    \multirow{3}{*}{\makecell{Noise \\ (1000)}}
        & GECM-HTEM & 0 & 0 & 0 & 0 & 3 \\
        & SSQLASSO & 0 & 0 & 0 & 0 & 819 \\
        & EMVS & 0 & 0 & 0 & 0 & 2 \\
    \hline
\end{tabular}
\label{tab:app_real_boston_vs}
\end{table}

\section{Some Useful Distributions} \label{app:densities}

Some of the frequently encountered probability distributions in this article are reviewed in this section, to remove ambiguity about the parameterizations of their density functions used herein.

\begin{enumerate}

\item The density function of a hyperbolic distribution $\mathrm{Hyperbolic}(\eta, \rho^2)$ is
\begin{equation}
    f_{\mathrm{Hyperbolic}}(x | \eta, \rho^2) = \dfrac{1}{2 \sqrt{\eta \rho^2} K_1(\eta)} e^{-\left\{ \eta(\eta + x^2/\rho^2) \right\}^{1/2}},
    \quad -\infty < x < \infty,
\label{eq:hyperbolicdensity}
\end{equation}
with shape parameter $\eta > 0$ and scale parameter $\rho^2 > 0$, and $K_1$ as a modified Bessel function of the second kind.

\item The density  function of a Student-$t$ distribution $\mathrm{t}(\eta, \rho^2)$ with location parameter 0, scale parameter $\rho^2 > 0$ and $\eta > 0$ degrees of freedom is
\begin{equation}
    f_{\mathrm{t}}(x | \eta, \rho^2) = \dfrac{\Gamma((\eta+1)/2)}{\Gamma(\eta/2)} \dfrac{\eta^{\eta/2}}{\sqrt{\pi\rho^2}} \left(\eta + \dfrac{x^2}{\rho^2}\right)^{-(\eta+1)/2},
    \quad -\infty < x < \infty.
\label{eq:tdensity}
\end{equation}

\item
The generalized inverse Gaussian (GIG) distribution, denoted by $\mathrm{GIG}(\lambda, a, b)$, has a density
\begin{equation}
    f_{\mathrm{GIG}}(x \mid \lambda, a, b) = \dfrac{(a/b)^{\lambda/2}}{2 K_{\lambda}(\sqrt{ab})} x^{\lambda-1} e^{-(ax + b/x)/2},
    \quad x > 0,
\label{eq:gigdensity}
\end{equation}
where $a, b > 0$, $p \in \mathbb{R}$ and $K_{\lambda}$ is a modified Bessel function of the second kind.

\item
The inverse gamma distribution $\mathrm{InvGamma}(a, b)$ has a density
\begin{equation}
    f_{\mathrm{InvGamma}}(x \mid a, b) = \dfrac{b^a}{\Gamma(a)} x^{-a-1} e^{-b/x},
    \quad x > 0,
\label{eq:invgammadensity}
\end{equation}
where $a > 0$ and $b > 0$ respectively denote the shape and the scale parameters.

\item
The probability mass function of the discrete uniform distribution $\mathrm{DiscUniform}(\mathcal{G})$ on a finite set of real values $\mathcal{G}$ can be expressed as
\begin{equation}
    f_{\mathrm{DiscUniform}}(x) = \dfrac{1}{|\mathcal{G}|},
    \quad x \in \mathcal{G},
\label{eq:discuniformdensity}
\end{equation}
where $|\mathcal{G}|$ represents the number of elements in $\mathcal{G}$.

\end{enumerate}

\end{appendices}

\end{document}